\documentclass{article}
\usepackage[utf8]{inputenc}

\usepackage{graphicx} 
\usepackage{amsmath}
\usepackage{amsthm}
\usepackage{thmtools}
\usepackage{amssymb}
\usepackage[dvipsnames]{xcolor}
\usepackage{algorithm}
\usepackage{algpseudocode}
\usepackage{dsfont}
\usepackage{booktabs}
\usepackage{multirow}
\usepackage{colortbl}
\usepackage[textwidth=2cm]{todonotes}
\usepackage[a4paper,top=3cm,bottom=3cm,left=3cm,right=3cm,marginparwidth=4cm]{geometry}
\usepackage[hidelinks,colorlinks=true,linkcolor=blue,citecolor=blue]{hyperref}
\usepackage[nameinlink]{cleveref}
\usepackage{authblk}
\usepackage{csquotes}

\newtheorem{theorem}{Theorem}
\declaretheorem[name=Theorem,numbered=no]{theorem*}
\newtheorem{lemma}{Lemma}
\newtheorem{definition}{Definition}
\newtheorem{corollary}{Corollary}

\newtheorem{observation}{Observation}
\newtheorem{proposition}{Proposition}
\declaretheorem[name=Proposition,numbered=no]{proposition*}

\newcommand\leo[1]{\textcolor{black}{#1}}

\newcommand\chris[1]{\textcolor{black}{#1}}
\newcommand\yukinew[1]{\textcolor{black}{#1}}

\crefname{chapter}{Section}{Sections}

\newcommand\indeg[1]{\delta^-(#1)} 
\newcommand\net{\mathcal{N}} 
\newcommand\pth{\rightsquigarrow} 
\newcommand\mux{\bar\mu}

\newcommand\pair{$(b,a)$}

\usepackage[english]{babel}

\title{Metrics for classes of semi-binary phylogenetic networks using $\mu$-representations\footnote{This paper received funding from the Netherlands Organisation for Scientific Research (NWO) under
projects OCENW.M.21.306 and OCENW.KLEIN.125.}}
\author[1]{Christopher Reichling}
\author[1]{Leo van Iersel} 
\author[1]{Yukihiro Murakami\footnote{Corresponding author, email: y.murakami@tudelft.nl}}
\affil[1]{Delft Institute of Applied Mathematics, Delft University of Technology, The Netherlands}
\date{\today}

\begin{document}

\maketitle

\begin{abstract}
    Phylogenetic networks are 
    useful in representing the evolutionary history of taxa. 
    In certain scenarios, one requires a way to compare different networks.
    In practice, this can be rather difficult, except within specific classes of networks.
    In this paper, we derive metrics for the class of \emph{orchard networks} and the class of \emph{strongly reticulation-visible} networks, from variants of so-called \emph{$\mu$-representations}, which are vector representations of networks.
    For both network classes, we impose degree constraints on the vertices, by considering \emph{semi-binary} networks.
\end{abstract}

\section{Introduction}


Phylogenetic trees are used to model the evolutionary history of species~\cite{nei2000molecular}.
Recent studies have focused on generalizing trees to account for more complex evolutionary scenarios.
Trees suffice to illustrate vertical descent evolution, however fail to accommodate for reticulate evolution, which arises from hybridization events and horizontal gene transfers~\cite{bai2021defining}.
To represent such events, phylogenetic networks have proven to be more fruitful~\cite{huson2010phylogenetic,bapteste2013networks}.


How does one construct phylogenetic networks?
Traditional phylogenetic inference methods can be grouped into model-based methods (e.g. Bayesian inference~\cite{wen2016bayesian,huelsenbeck2001bayesian}, maximum likelihood~\cite{wen2018inferring,brocchieri2001phylogenetic}) or non-model-based methods (e.g. distance-based~\cite{bryant2004neighbor}, maximum parsimony~\cite{swofford1998phylogenetic}, and combinatorial~\cite{huson2011survey}).
In all cases, one must evaluate the accuracy of the output.
\yukinew{For certain evolutionary histories, the true phylogeny is known.
In validating the inference method, this means one can compare output phylogenies to the benchmark phylogeny.
In doing so, one needs a notion of computing distances between the two phylogenies.}
Existing metrics such as rearrangement moves suffer from computational intractability~\cite{DasGupta1997,bordewich2017lost,janssen2018exploring}; others like the triplet distance~\cite{dobson1975comparing} suffer from non-identifiability (two distinct networks could be at a distance 0. See e.g., Figure 19 in~\cite{cardona2008metrics}).

One way of avoiding these situations is to first find complete graph invariants, sometimes called encodings, for specific classes of phylogenetic networks.
Statements of the sort `two networks are isomorphic if and only if they have the same encodings' are typically sought after in this area.
Taking the symmetric difference 
of the invariants often leads to a metric, by definition of complete graph invariants~\cite{gurevich2001}.
In this paper, we consider invariants based on so-called $\mu$-vectors.
For every vertex in the graph, these vectors encode the number of paths from it to every leaf.
The multiset of all $\mu$-vectors is called the $\mu$-representation of the network. 
Cardona et al. introduced the $\mu$-representation and showed that it can be used as a metric for binary tree-child networks \cite{cardona2008comparison}.

In general, two networks can have the same $\mu$-representations; in this sense, we say that the $\mu$-representations do not encode networks, see \autoref{fig: counterYuki}. However, $\mu$-representations may encode subclasses of phylogenetic networks such as the class of tree-child networks that Cardona et al. considered as mentioned. \bigskip

\begin{figure}
    \centering
    \includegraphics[width=0.6\textwidth]{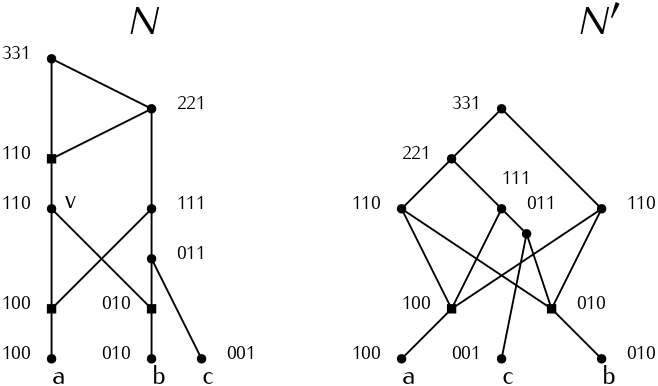}
    \caption{A binary stack-free network~$N$ and a semi-binary stack-free network~$N'$ with the same~$\mu$-representations ($\{331,221,111,110,110,011,100,100,010,010,001\}$). Note that the edges are directed downward. For instance node $v$ has $\mu$-vector $110$ because from it there is a single directed path to leaf $a$, a single directed path to leaf $b$ and there are no directed paths to leaf $c$.} 
    \label{fig: counterYuki}
\end{figure}
For most mathematical and algorithmic techniques, the full class of phylogenetic networks is too large. Therefore, several restricted classes of phylogenetic networks have been defined and studied. A network is called \emph{tree-child} if none of its nodes have only reticulation children and a network is \emph{stack-free} if no two reticulations are adjacent. A network is \emph{reticulation-visible} if for each reticulation there exists a leaf such that all paths to this leaf visit the reticulation. These classes have mainly been defined for their nice mathematical properties. However, an intuitive biological argument can be made for tree-child networks as well. As long as a species does not go extinct it is highly unlikely that all of its surviving offspring is the result of hybridization. Tree-child networks are automatically stack-free, because if two reticulations are adjacent then one of them must have the other one as their only child. Reticulation-visible networks are stack-free as well. Moreover, any phylogenetic network can be made stack-free by iteratively identifying any two adjacent reticulations. More recently the class of \emph{orchard} networks was introduced as a superclass of the class of tree-child networks with nice characteristics. A natural justification for this class is that orchard networks can be interpreted as trees with additional \emph{horizontal} arcs which correspond to horizontal gene transfer \cite{van2022orchard}. For example the left network in  \autoref{fig: counterYuki} is orchard and \autoref{fig: horizontalArcs} shows how it can be represented as a tree with additional horizontal arcs.

\begin{figure}
    \centering
    \includegraphics[width=0.6\textwidth]{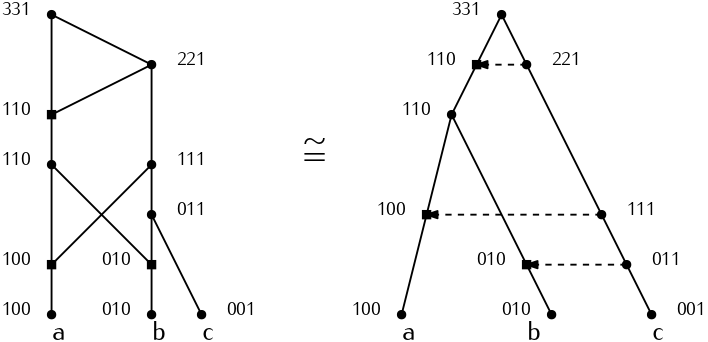}
    \caption{An orchard network can be represented as a tree with additional horizontal arcs.} 
    \label{fig: horizontalArcs}
\end{figure}

Building upon the work of Cardona et al., Erd\H{o}s et al. sought to extend the application of $\mu$-representations to a larger class of phylogenetic networks in \cite{erdHos2019class}, which they \yukinew{called} orchard networks. 
\yukinew{Their proof was based on identifying and reducing so-called \emph{cherries} and \emph{reticulated cherries}, straight from the $\mu$-representations.}
However, some of their findings were later refuted by Bai et al. 
\cite{bai2021defining} 
\leo{who}
showed that it is not possible to determine reticulated cherries \yukinew{from} the $\mu$-representation for general binary orchard networks. In that paper, Bai et al. proposed a \emph{stack-freeness} constraint within the class of orchard networks to establish the encoding result. They then aimed to show encoding holds for semi-binary stack-free orchard networks (because binary networks can be made stack-free by identifying stacks of reticulations, \leo{which} 
makes the network semi-binary).
\yukinew{This claim was shown, by counter-example, to only hold for networks which are binary \cite{murakami2021phylogenetic}}.
In \cite{cardona2024comparison}, Cardona et al. proposed an extended $\mu$-representation, which also takes into consideration the number of paths to reticulations from each node. 
In the paper they showed that this extended $\mu$-representation is an encoding for binary orchard networks, lifting the stack-free condition. 
This modification however does not show encoding for semi-binary \leo{orchard} networks as originally proposed in \cite{bai2021defining}, as the proof is restricted to networks which are binary. 

\yukinew{As it stands, $\mu$-representations encode stack-free binary orchard networks, and the extended $\mu$-representations encode binary orchard networks.}
\leo{The aims of this paper are to find a large subclass of semi-binary orchard networks that is encoded by $\mu$-representations and to find variants of $\mu$-representations 
that form an encoding for all semi-binary orchard networks.}
Our contributions are as follows, in \Cref{sec: extended mu} we first propose a modified $\mu$-representation including the in-degrees of nodes, which is different from the extended $\mu$-representation proposed by Cardona et al. in \cite{cardona2024comparison}. \autoref{thm: orchard encoding} states that this modified $\mu$-representation encodes semi-binary stack-free orchard networks. 
With this theorem, we can define a metric given by the cardinality of the symmetric difference of the modified $\mu$-representations. \leo{On the other hand, we}
show that encoding does not hold for non-binary stack-free orchard networks even if the out-degrees are also 
\leo{added to the modified $\mu$-representation}
(\autoref{thm: non-binary orchard not encoded}).

Furthermore, in \Cref{sec: strongly ret-vis encoding} we present a fundamental equation which governs the relationship between the in-degrees of reticulations and the $\mu$-representation of a network. We prove that such an equation exists (\autoref{thm: indeg formula}), and show how this gives rise to a system of equations on the $\mu$-vector of the root and the $\mu$-vectors and in-degrees of reticulations. We furthermore show, that for reticulation-visible networks with fixed reticulation set, the system of equations generated by \autoref{thm: indeg formula} has a unique solution (\autoref{thm: indeg ret vis}).
Then, we define a new class of networks called \emph{strongly reticulation-visible}
networks, for which there is a \emph{tree-path} \leo{(a path containing only tree-nodes, which may consist of just a single tree-node)} 
to a bridge from each child of a reticulation. A \emph{bridge} is an edge which disconnects the network if cut. We show that a bridge and the lowest reticulation ancestor of that bridge in any network \leo{are} 
uniquely determined by the $\mu$-representation (\autoref{thm: bridge-node mu-rep} and \autoref{lem: lowest ret above bridge}). We then use this to show that strongly reticulation-visible networks with the same $\mu$-representations have 
the same modified $\mu$-representation (\autoref{thm: strong ret vis in-degrees}). Finally\leo{,} we conclude that strongly reticulation-visible semi-binary stack-free orchard networks are encoded in the space of semi-binary stack-free networks by the $\mu$-representation (\autoref{thm: strong ret vis encoding}). This means that the cardinality of the symmetric difference of the $\mu$-representations gives a metric between these networks. \leo{Hence, it is not necessary to include indegrees in the $\mu$-representation for this class.}

\section{Preliminaries} \label{sec: preliminaries}
By a \emph{rooted directed acyclic graph (rooted DAG)} we mean a connected directed graph that has no directed cycles. 
A rooted DAG whose leaves are bijectively labeled by the elements of a finite set $X$, we call an \emph{$X$-DAG}. 
\yukinew{We assume henceforth that there is some ordering on the elements of $X$, i.e., that $X=\{x_1,\ldots, x_n\}$. This helps with defining certain terms.}
The in-degree of a node $v$ in an~$X$-DAG, which we denote $\indeg{v}$, is the number of edges which end in $v$. 
The out-degree $\delta^+(v)$ of a node $v$ 
is the number of edges starting in $v$. The degree of a node is the sum of its in-degree and out-degree. Every vertex of an $X$-DAG can be classified by their in-degree or their out-degree.
In particular, an~$X$-DAG contains
\begin{itemize}
    \item a single \emph{root} $\rho$ with in-degree $\indeg{\rho} = 0$;
    \item \emph{tree-nodes} $v$ with in-degree $\indeg{v} \leq 1$;
    \item \emph{reticulations} $r$ with in-degree $\indeg{r} \geq 2$; 
    \item and \emph{leaves} $a$ with out-degree $\delta^+(a) = 0$.
\end{itemize}
Note that the root is a tree-node, and leaves can be either tree-nodes or reticulations. Nodes which are not leaves are sometimes called \emph{internal} nodes.
A node with indegree-1 and outdegree-1 is called an \emph{elementary node}. 
The set of all reticulations contained in a given $X$-DAG $\net$ we will denote $R(\net)$ or simply $R$ when the $X$-DAG is obvious from the context. 
We call an edge a \emph{reticulation edge} if it ends in a reticulation. 
The \emph{hybridization number} $h(\net) = \sum_{r_i \in R}(\indeg{r_i} - 1)$ of an $X$-DAG $\net$ is the number of reticulation edges minus the number of reticulations. 
We call an $X$-DAG a \emph{tree} if it has no reticulations.

A \emph{phylogenetic network} is an $X$-DAG without parallel arcs or elementary nodes, where the root must be a leaf (in which case the network is one on a single leaf) or have out-degree greater than or equal to 2, reticulations have out-degree 1 and leaves are tree-nodes. A phylogenetic network in which all nodes except the root or the leaves have degree 3 is called \emph{binary}. A phylogenetic network in which all tree-nodes except the root or the leaves have degree 3 but reticulations can have degree $\geq 3$ is called \emph{semi-binary}. 
\yukinew{Biologically, such evolutionary histories can occur if there are ambiguities in the order of consecutive reticulate events.}
We say a phylogenetic network is \emph{non-binary} when there are no such added restrictions on the degrees of the nodes. 
Note that non-binary does not mean that the network is `not binary', but rather, that the network is `not necessarily binary'.
A phylogenetic network which does not contain any reticulations is called a \emph{phylogenetic tree}. See \autoref{fig: example nets} for some examples. Henceforth, we may refer to phylogenetic networks as networks for brevity.

All phylogenetic networks are \emph{$X$-DAGs}. From here on out we will identify the leaf nodes with the elements of the set $X$ and no longer make a distinction between the two. Furthermore, we will assume edges are directed unless otherwise mentioned, and in all figures edges will be directed downward, such that the root is at the top and the leaves are at the bottom.


A path $v_0 \pth v_k$ between two nodes $v_0,v_k \in V$ is a sequence of edges $v_0v_1,v_1v_2,...,v_{k-1}v_k$ such that $v_iv_{i+1}\in E$ for $i \in\{0,1,\ldots,k-1\}$. Note, that this means that in this paper all paths will be directed paths as all edges are directed.


We say a node $v_1$ is an \emph{ancestor} of another node $v_2$ if there is a path from $v_1$ to $v_2$. We also say that $v_2$ is a \emph{descendant} of $v_1$. In this case we may also say $v_1$ is \emph{above} $v_2$ and that $v_2$ is \emph{below} $v_1$. If the path consists of a single edge, then we say $v_1$ is the \emph{parent} of $v_2$, usually denoted $p_{v_2}$ and $v_2$ is the \emph{child} of $v_1$. We also consider the trivial path, therefore each node is both an ancestor and a descendant of itself. The number of paths from $v_1$ to $v_2$ we will denote $P_{v_1v_2}$.

Given a directed edge $e = v_1v_2$ we call $v_2$ the \emph{head} of $e$ and $v_1$ the \emph{tail} of $e$. We say a node is below $e$ if it is a descendant of $v_2$ and we say it is above $e$ if it is an ancestor of $v_1$. We say two nodes are connected if there is an undirected path between them. We say a set of nodes is connected if every pair of nodes in the set is connected. We say a graph is connected if the set of its vertices is connected. Recall that an $X$-DAG is a connected graph.

A \emph{tree-path} is a path $v_0 \pth v_k$, such that $v_i$ is a tree-node for each $i \in \{0,1,\ldots,k\}$. A tree-node $v$ which has out-degree $\delta^+(v) = 1$ we shall call an \emph{elementary} node. A path for which all but the start and end nodes are elementary nodes, we shall call an \emph{elementary path}. The \emph{height} of a node is the length of the longest path from the node to a leaf.

Two $X$-DAGs $\net = (V,E)$ and $\net' = (V',E')$ are said to be \emph{isomorphic}, denoted by $\net \cong \net'$, when there exists a bijective function $f: V \rightarrow V'$ such that $f(a) = a$ for all $a \in X$ and $v_1v_2 \in E \iff f(v_1)f(v_2) \in E'$ for all $v_1,v_2 \in V$.

\subsection{Network classes}

\begin{figure}
    \centering
    \includegraphics[width = \textwidth]{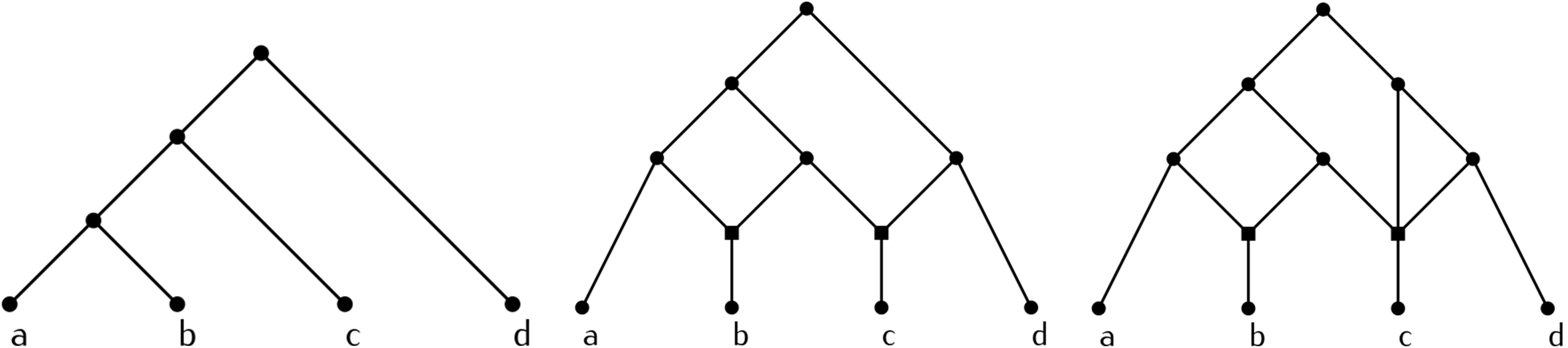}
    \caption{A binary phylogenetic tree, a binary phylogenetic network and a semi-binary phylogenetic network.}
    \label{fig: example nets}
\end{figure}
A phylogenetic network $\net$ is said to be \emph{stack-free} if no reticulation in $\net$ is the child of another reticulation. A phylogenetic network such that for each reticulation there is a leaf for which all paths from the root to this leaf pass through the reticulation is called \emph{reticulation-visible}. The network in \autoref{fig: mu-example} is binary and stack-free. However, it is not reticulation-visible, because there are no leaves such that all paths from the root pass through $u$. All networks in \autoref{fig: example nets} are stack-free and reticulation-visible. In \Cref{sec: mu stabilitiy}, we will introduce the class of  \emph{strongly reticulation-visible networks} as the class of phylogenetic networks, in which there is a tree-path to a bridge from the child of each reticulation.

To define the class of \emph{orchard networks}, we require notions of \emph{cherries}, \emph{reticulated cherries}, and their reductions. A \emph{cherry} is an ordered pair of leaves \pair{} which have the same parent. A \emph{reticulated cherry} is an ordered pair of leaves \pair{} such that the parent $p_b$ of $b$ is a reticulation and the parent of $a$ is a tree-node $p_a$ which is also the parent of $p_b$. A pair \pair{} which is either a cherry or a reticulated cherry is also called a \emph{reducible pair}. \emph{Suppressing} an elementary node is the action of deleting the node and adding an edge between the parent and the child of the node. To \emph{reduce} a cherry in a network $\net$, we delete the leaf $b$ and suppress its parent $p_b$ if it has become elementary. To reduce a reticulated cherry in $\net$ we delete the edge $p_ap_b$ and suppress any nodes which have become elementary. In this way one always obtains another phylogenetic network as the result of reducing a reducible pair in a phylogenetic network.

A network is called \emph{orchard} if there exists a sequence $s_1s_2s_3\ldots s_i \dots s_n$, of ordered pairs, such that $s_i$ is a reducible pair in the network after reducing each pair in the sequence up to $s_{i-1}$ and the entire sequence reduces the network to a network on a single leaf. Note that in that case, each network generated by performing reductions $s_1$ up to $s_i$, is orchard with sequence $s_{i+1},s_{i+2},\ldots,s_n$, see Corollary 4.2 in \cite{erdHos2019class}. \chris{Furthermore, any maximal sequence of cherry reductions is complete, meaning any cherry reduction will result in an orchard network and any partial sequence can be made complete}. The network in \autoref{fig: mu-example} is orchard with sequence $(b,c)(a,c)(b,a)(a,c)(c,a)$. It contains the reticulated cherry $(b,c)$. The networks in \autoref{fig: mu-example} are all orchard and contain the reducible pair \pair{}. In the phylogenetic tree \pair{} is a cherry, while in the other networks \pair{} is a reticulated cherry. See \autoref{fig: classes of phylogenetic networks} for a visualization of the way the different classes of phylogenetic networks discussed in this paper are related.
\begin{figure}
    \centering
    \includegraphics[width = 0.3\textwidth]{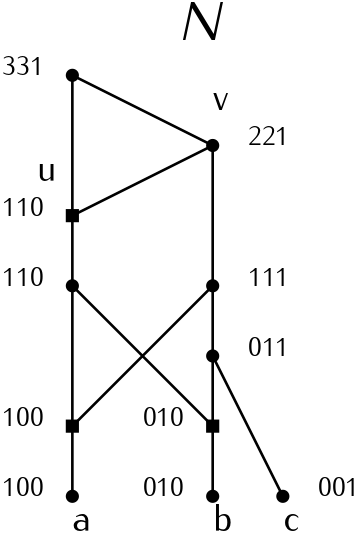}
    \caption{An example network $\net$, on leaves $a,b$ and $c$, with $\mu$-representation $\mu(\net) = \{(100,2),(010,2),(001,1),(110,2),(011,1),(111,1),(221,1),(331,1)\}$. Edges are directed downward. Nodes $u$ and $v$ have vectors $\mu(u) = 110$ with multiplicity 2 and $\mu(v) = 221$ with multiplicity 1. Following \Cref{def:modifiedmu}, we have~$\mux(u) = (1,110)$ and $\mux(v) = (1,221)$.}
    \label{fig: mu-example}
\end{figure}

\begin{figure}
    \centering
    \includegraphics[width= 0.6\textwidth]{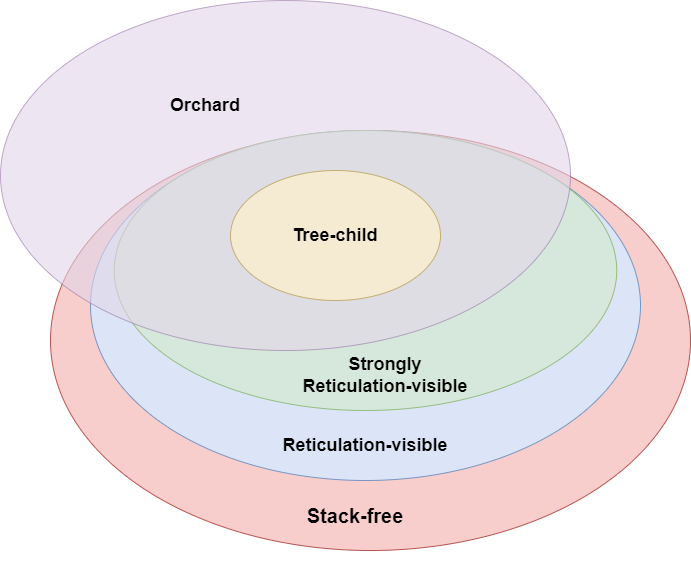}
    \caption{A diagram showing the relations between several different classes of phylogenetic networks.}
    \label{fig: classes of phylogenetic networks}
\end{figure}

\subsection{The \texorpdfstring{$\mu$}{m}-representation}
Given an $X$-DAG $\net = (V,E)$, the \emph{path-multiplicity vector} or $\mu$\emph{-vector} of any node $v\in V$ is defined as follows: let $\mu(v) \in \mathbb{Z}^X$ be a vector indexed by the leaves of $\net$, such that the coordinate indexed by leaf $a$, denoted $\mu(v)_a$, is equal to the number of paths from $v$ to $a$. Note that $\mu(v)$ only contains non-negative integer coordinates and is never equal to the zero vector. Moreover, since any leaf constitutes a trivial path as well their $\mu$-vector is the
corresponding unit vector. With the exception of leaf nodes, the $\mu$-vector of a node is always the sum of the $\mu$-vectors of its children. For two vertices~$u,v\in V$, we write~$\mu(u)\le\mu(v)$ if $\mu(u)_a\le\mu(v)_a$ for all~$a\in X$. We write $\mu(u) < \mu(v)$ if~$\mu(u)\le \mu(v)$ and~$\mu(u)_a<\mu(v)_a$ for at least one~$a\in X$.
Note that if~$u$ is descendant of~$v$, it always holds that~$\mu(u) \le \mu(v)$. We say that a $\mu$-vector~$\mu(v)$ \emph{belongs} to a node~$u$ if~$\mu(u) = \mu(v)$. In particular,~$\mu(v)$ belongs to~$v$. Later on, we shall see that a $\mu$-vector can belong to multiple nodes. The $\mu$-representation of $\net$, denoted $\mu(\net)$, is the multiset of all $\mu$-vectors of nodes in $V$.

A multiset is similar to a set. The main difference between a multiset and a set, is that a multiset can contain multiple instances of the same element. The number of instances of an element in a given multiset is called the multiplicity of that element in that multiset. For example, if the $\mu$-representation $\mu(\net)$ contains two instances of a vector $\mu(v)$, then we say $\mu(v)$ has multiplicity 2 in $\mu(\net)$. We may shorten this to $\#\mu(v) = 2$, whenever the multiset containing $\mu(v)$ is implied. Usually the implied multiset is $\mu(\net)$. Then, $\#\mu(v)$ denotes the multiplicity of $\mu(v)$ in $\mu(\net)$.

If a given $\mu$-vector $\mu(v)$ is not contained in a multiset $\mu(\net)$, we may say $\mu(v)$ has multiplicity 0 in $\mu(\net)$. The operation of removing a $\mu$-vector $\mu(v)$ from $\mu(\net)$ is equivalent to lowering the multiplicity of $\mu(v)$ in $\mu(\net)$ by 1. Clearly, the multiplicity of a $\mu$-vector cannot be negative and a $\mu$-vector which has multiplicity 0 in $\mu(\net)$ cannot be removed from $\mu(\net)$. The operation of adding a $\mu$-vector $\mu(v)$ to a multiset $\mu(\net)$ is equivalent to increasing the multiplicity of $\mu(v)$ in $\mu(\net)$ by 1. We say that $\mu(\net)$ is generated by adding $\mu(v)$ for each node $v \in V$. Therefore, the multiplicity of a vector $\mu(v)$ in $\mu(\net)$ is equal to the number of nodes in $\net$ with $\mu$-vector equal to $\mu(v)$. We do not equate the nodes $v \in V$ with their $\mu$-vectors because multiple nodes may have the same $\mu$-vector. A set of at least two tree-nodes which have the same $\mu$-vector we shall call \emph{tree-clones}. A node which is part of a set of tree-clones, we shall call a tree-clone. Bai et al. showed in \cite{bai2021defining} Lemma 4.4, that orchard networks do not contain tree-clones.

For example, the node $u$ in \autoref{fig: mu-example} has $\mu$-vector $110$, because there is exactly one path to leaf $a$, one path to leaf $b$, and there are no paths to leaf $c$ starting in $u$. There are two instances of nodes with $\mu$-vector $110$, because the paths starting in the reticulation $u$ are in bijection with the paths starting in its child, by adding or deleting the edge between them. Therefore, $\mu(u)$ has multiplicity 2 in $\mu(\net)$. The node $v$ in \autoref{fig: mu-example} has $\mu$-vector $221$, because there are 2 paths to leaf $a$, one via node $u$ and one via the other child of $v$, and 2 paths to $b$ and one path to $c$, starting in $v$. It should be clear from these examples why, with the exception of leaf nodes, the $\mu$-vector of a node is always the sum of the $\mu$-vectors of its children. 



\subsection{The symmetric difference of multisets}
The symmetric difference between two sets $S_1, S_2$ is the set of elements from $S_1$ and $S_2$, which are not contained in both sets.
\begin{equation*}
    S_1 \triangle S_2 = (S_1 \cup S_2) \setminus (S_1 \cap S_2)
\end{equation*}
The cardinality of the symmetric difference, or the number of elements that are unique to either set, can be used as a measure for the difference between these two sets. 

For multisets $M_1,M_2$,
the \emph{symmetric difference} $M_1\triangle M_2$ contains elements in~$M_1$ or~$M_2$, with multiplicity equal to the absolute difference of the respective multiplicities within~$M_1$ or~$M_2$. The cardinality of a multiset is the sum of the multiplicities of all its elements.
The cardinality of the symmetric difference is a metric on multisets.
For it to be a metric on some subset of phylogenetic networks, we need a modified version of the \chris{identity axiom} to hold. Let~$\net_1$ and~$\net_2$ be networks.
\begin{itemize}
    \item $|\mu(\net_1) \triangle \mu(\net_2)| = 0$ if, and only if $\net_1$ is isomorphic to $\net_2$.
\end{itemize}
Therefore, we will be focusing on determining the conditions such that the $\mu$-representations are equal if, and only if the networks are isomorphic.
It is important to note that there is always only one $\mu$-representation belonging to a given network.







\section{Encoding by 
modified \texorpdfstring{$\mu$}{m}-representations} 
\label{sec: extended mu}

\subsection{Preliminary lemmas}
In this section we will show one of the main results. This result is in a way a continuation and modification of previous propositions by Bai et al. \cite{bai2021defining} and Erd\H{o}s et al. \cite{erdHos2019class}.
We will make use of a modified $\mu$-representation.
Let~$a=(a_1,\ldots, a_k)$ and~$b = (b_1,\ldots, b_\ell)$ be vectors. 
We write~$a\oplus b = (a_1,\ldots, a_k, b_1, \ldots, b_\ell)$ to denote the concatenation of the two vectors.
If one of the arguments (either $a$ or $b$) is an integer, it is understood and treated as a 1-dimensional vector.

\begin{definition}\label{def:modifiedmu}
    Let $\net = (V,E)$ be a network.
    Given a $\mu$-vector $\mu(v)$ of a node $v \in V$, its \emph{modified $\mu$-vector} is
    \begin{equation*}
        \mux(v) = \indeg{v} \oplus \mu(v).
    \end{equation*}
    For a network on leaf set~$X$, $\mux(v)$ is the~$|X|+1$-th dimensional vector with~$\indeg{v}$ as the first coordinate and the coordinates of~$\mu(v)$ as the latter~$|X|$ coordinates.
    The \emph{modified $\mu$-representation} $\mux(\net)$ of a network $\net$, is the multiset of modified $\mu$-vectors $\mux(v)$ of nodes $v$ in $\net$.
\end{definition}

We also define two types of reticulated cherries.

\begin{definition} 
    Let~$\net$ be a network and let~$\pair{}$ be a reticulated cherry, where~$p_b$ is the parent of~$b$.
    \begin{itemize}
        \item If $\indeg{p_b} = 2$, then \pair{} is a \emph{simple reticulated cherry}.
        \item If $\indeg{p_b} \ge 3$, then \pair{} is a \emph{complex reticulated cherry}.
    \end{itemize}
\end{definition}

See \autoref{fig: types of cherries}, for examples of a cherry, a simple reticulated cherry and a complex reticulated cherry.

\begin{figure}[ht]
    \centering
    \includegraphics[width = 0.8\textwidth]{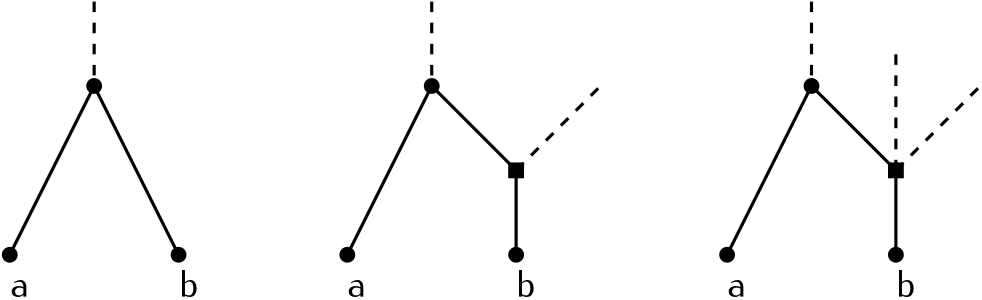}
    \caption{A cherry, a simple reticulated cherry~\pair{} and a complex reticulated cherry~\pair{} on the leaves $a$ and $b$.}
    \label{fig: types of cherries}
\end{figure}

For the rest of the section, let~$\net$ be a semi-binary stack-free network on~$X$,
and let~$a,b \in X$ be leaves of $\net$. We shall show encoding results regarding modified $\mu$-representations. First, we prove some preliminary results.

\begin{observation} \label{lem: stack free leaf mults} \label{lem: leaf tree parent}
    Let $a$ be a leaf in~$\net$. Then, $\mu(a)$ has multiplicity 1 or 2 in $\mu(\net)$. If $\#\mu(a) = 1$, then its parent $p_a$ is a tree-node with $\mu(p_a) \ne \mu(a)$. Furthermore, $\mu(p_a)$ is minimal in the multiset $\{\mu(v):\mu(v) > \mu(a), v\in V\}$. Otherwise, if~$\#\mu(a) = 2$, then~$p_a$ is a reticulation with $\mu(p_a) = \mu(a)$.
\end{observation} 

Now we will show that cherries and reticulated cherries are uniquely determined by $\mu(\net)$. 

\begin{lemma}\label{lem: cherry mu}
The pair \pair{} is a cherry in $\net$ if, and only if, $\mu(v)_a = \mu(v)_b$ for each~$\mu(v) \in \mu(\net)\setminus\{\mu(a),\mu(b)\}$.
\end{lemma}

\begin{proof}
Let us assume first that the pair \pair{} is a cherry in $\net$ with parent node $p$. Then, for each node $v \in V\setminus \{a,b\}$ the number of paths from $v$ to either $a$ or $b$ is equal to the number of paths~$P_{vp}$ from $v$ to $p$, so $\mu(v)_a = \mu(v)_b$.\medskip

For the other direction, we will use a proof by contradiction. Assume $\mu(v)_a = \mu(v)_b$ for each $\mu(v) \in \mu(\net)\setminus\{\mu(a),\mu(b)\}$. Now assume the pair \pair{} is not a cherry in $\net$. This means $a$ and $b$ must have different parents $p_a \neq p_b$. However, because $\mu(v)_a = \mu(v)_b$ for each $\mu(v) \in \mu(\net)\setminus\{\mu(a),\mu(b)\}$, we have that $\mu(p_a)_b = \mu(p_a)_a = 1$, therefore there is a path from $p_a$ to $b$. This means $p_a$ must be an ancestor of $p_b$. But also $\mu(p_b)_a = \mu(p_b)_b = 1$, therefore $p_b$ must also be an ancestor of $p_a$. In acyclic graphs two nodes cannot be ancestors of each other unless they are the same node, therefore $p_a = p_b$, but this contradicts our assumption that \pair{} is not a cherry.
\end{proof}

Note that the condition on the $\mu$-vectors implies that $\mu(a)$ and $\mu(b)$ have multiplicity 1 in $\mu(\net)$, because if for instance $\mu(b)$ has multiplicity greater than 1 in $\mu(\net)$, then $\mu(\net)\setminus\{\mu(a),\mu(b)\}$ would still contain a vector $\mu(b)$, for which $\mu(b)_a = 0 \neq 1 = \mu(b)_b$.

\begin{lemma} \label{lem: ret cherry mu}
  The pair \pair{} is a reticulated cherry in $\net$ with $b$ the reticulation leaf if, and only if, $\#\mu(a) = 1$, $\#\mu(b)=2$, $\mu(v)_b \geq \mu(v)_a$ for each $\mu(v) \in \mu(\net)\setminus\{\mu(a),\mu(b)\}$ and $\mu(\net)$ contains a vector $\mu(p_a) = \mu(a) + \mu(b)$.
\end{lemma}
\begin{proof}
First let us assume \pair{} is a reticulated cherry in $\net$ with $b$ the reticulation leaf. Then the parent of $a$ is a tree node $p_a$ and, by \autoref{lem: stack free leaf mults}, $\mu(a)$ has multiplicity 1 in the multiset. Also, the parent of $b$ is a reticulation $p_b$, therefore by \autoref{lem: stack free leaf mults}, $\mu(b)$ has multiplicity 2 in the multiset. Furthermore, $p_a$ is the parent of $p_b$ and thus an ancestor of $b$. Therefore, for each path $v \pth p_a$ with $v \in V\setminus \{a,b\}$, there is at least one path $v \pth b$ via $p_a$. Furthermore, the number of paths from $v$ to $a$ equals the number of paths from $v$ to $p_a$. This means that $\mu(v)_b \geq P_{vp_a} = \mu(v)_a$ for any node $v \in V\setminus \{a,b\}$. Finally, note that $\mu(p_a) = \mu(a) + \mu(p_b) = \mu(a) + \mu(b)$. This proves the first direction. \medskip

For the second direction, we will use proof by contradiction. Let us assume, $\mu(a)$ has multiplicity 1 in the multiset, $\mu(b)$ has multiplicity 2 in the multiset, $\mu(v)_b \geq \mu(v)_a$ for each $\mu(v) \in \mu(\net)\setminus\{\mu(a),\mu(b)\}$ and $\mu(\net)$ contains $\mu(p_a) = \mu(a) + \mu(b)$. Now assume \pair{} is not a reticulated cherry. Note $\mu(p_a) > \mu(a)$ and the only $\mu$-vectors $\mu(v)$ with $\mu(v) < \mu(p_a)$ are $\mu(a)$ and $\mu(b)$, thus $\mu(p_a)$ is minimal in $\{\mu(v) : \mu(v) > \mu(a), v \in V\}$. Then, by \autoref{lem: leaf tree parent}, $\mu(p_a)$ belongs to the parent of $a$. Furthermore, by \autoref{lem: stack free leaf mults}, we know $b$ has a reticulation parent $p_b$. Therefore the parent $p_a$ of $a$ is not a parent of $p_b$ the parent of $b$, because otherwise \pair{} would be a reticulated cherry. But there must be a path from $p_a$ to $b$ because $\mu(p_a)_b \geq \mu(p_a)_a = 1$. This means there must be other nodes on the path from $p_a$ to $p_b$. Let $c$ then be the child of $p_a$, then $\mu(p_a) = \mu(a) + \mu(c) = \mu(a) + \mu(b)$. Subtracting $\mu(a)$ gives $\mu(c) = \mu(b)$. This means $c$ is either the leaf $b$, which is not possible, or it is $p_b$, which we assumed it was not, or it is some other reticulation which has a child with the same $\mu$-vector equal to $\mu(b)$. Its child cannot be the leaf $b$, because by assumption $c$ is not $p_b$ and its child cannot be $p_b$ because the network is stack-free. Therefore, its child must be a tree-node with the same $\mu$-vector as the leaf $b$, which is not itself leaf $b$ or $p_b$. But, then $\#\mu(b) = 3$, which contradicts our assumption.
\end{proof}

Let us now define cherries and reticulated cherries as subsets of the original and modified $\mu$-representation as follows naturally from \Cref{lem: cherry mu} and \Cref{lem: ret cherry mu}.
\chris{
\begin{definition}\hfill
\begin{itemize}
    \item The pair \pair{} is a cherry in $\mu(\net)$ if $\mu(v)_a = \mu(v)_b$ for each $\mu(v) \in \mu(\net) \setminus \{\mu(a),\mu(b)\}$. 
    \item The pair \pair{} is a cherry in $\mux(\net)$ if $\mux(v)_a = \mux(v)_b$ for each $\mux(v) \in \mux(\net) \setminus \{\mux(a),\mux(b)\}$.
    \item The pair \pair{} is a reticulated cherry in $\mu(\net)$ with $b$ the reticulation leaf if $\#\mu(a) = 1$, $\#\mu(b)=2$, $\mu(v)_b \geq \mu(v)_a$ for each $\mu(v) \in \mu(\net)\setminus\{\mu(a),\mu(b)\}$ and $\mu(\net)$ contains a vector $\mu(p_a) = \mu(a) + \mu(b)$.
    \item The pair \pair{} is a reticulated cherry in $\mux(\net)$ with $b$ the reticulation leaf if the following hold:
    \begin{itemize}
        \item there does not exist $\mux(v) \in \mux(\net)$ with $\mux(v)_i = \mux(a)_i$ for $i \in X$,
        \item there exists a vector $\mux(p_b) \in \mux(\net)$ with $\mux(v)_i = \mux(b)_i$ for $i \in X$,
        \item $\mux(v)_b \geq \mux(v)_a$ for each $\mux(v) \in \mux(\net)\setminus\{\mux(a),\mux(b)\}$ and
        \item $\mux(\net)$ contains a vector $\mux(p_a)$ s.t. $\mux(p_a)_i = \mux(a)_i + \mux(b)_i$ for $i \in X$.
    \end{itemize}
    It is simple if $\mux(\net)$ contains $\mux(p_b) = [2] \oplus \mu(b)$ and otherwise it is complex.
\end{itemize}
\end{definition}
}
Note that if \pair{} is a cherry or a reticulated cherry in $\net$, it is a cherry or reticulated cherry in $\mu(\net)$ and $\mux(\net)$. Furthermore, if \pair{} is a reticulated cherry, it is simple if, and only if, $\mux(\net)$ contains $\mux(p_b) = [2] \oplus \mu(b)$, because then $\mux(p_b)_0 = 2$, for $p_b$ the parent of $b$. Otherwise, it is complex. If \pair{} is a cherry or a reticulated cherry in $\net$ we say that \pair{} is a reducible pair in $\net$, in $\mu(\net)$ and in $\mux(\net)$.\medskip

\begin{figure}
    \centering
    \includegraphics[width = 0.7\textwidth]{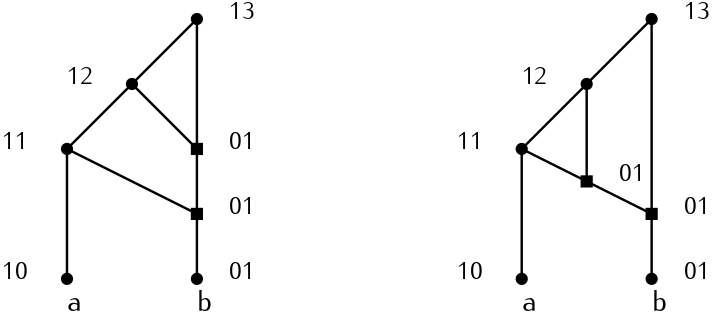}
    \caption{Two networks which are not stack-free. Although they have the same $\mu$-representation, they are not isomorphic. In the first network $(b,a)$ is a reticulated cherry while in the second network it is not. The first network is orchard, as it can be reduced by the sequence~$(b,a)(b,a)(b,a)$. The second network is not orchard.}
    \label{fig: stack-free counter}
\end{figure}

It is important to mention here that we restrict to the class of stack-free networks to be able to identify reticulated cherries in the $\mu$-representation.
Without the stack-free assumption, the conditions in \autoref{lem: ret cherry mu} are not sufficient to determine whether \pair{} is a reticulated cherry in general. In \autoref{fig: stack-free counter}, two non-isomorphic networks are displayed which have the same $\mu$-representation. To see the non-isomorphism, the left network contains a reticulated cherry while the right network does not. Clearly, $\mu$-representations do not suffice in identifying reticulated cherries if the stack-free condition is not imposed. 
In light of this, 
Cardona et al. proposed a different \emph{extended $\mu$-representation} 
by considering, for every vertex, the number of paths to reticulations, in addition to the $\mu$-vectors.
While this sufficed to prove that binary orchard networks are encoded by extended $\mu$-representations, the encoding result does not easily translate to the semi-binary network case.
The networks in \autoref{fig: counterYuki} form a counterexample, they have the same extended $\mu$-representation, while they are not isomorphic.

\subsection{Reconstructing orchard networks}
Now let us define cherry and reticulated cherry reductions in $\mux(\net)$. 
We first define operations on vectors.
Recall that to remove a vector from a multiset means to lower the multiplicity by 1 to a minimum of 0. If a vector has multiplicity 0 in a multiset we say the multiset does not contain the vector. 
\medskip

Let \pair{} be a cherry in $\net$. We define the cherry reduction of \pair{} in $\mux(\net)$ as the following operations:
\begin{enumerate}
    \item Remove \chris{the unit vector} $\mux(b)$ from $\mux(\net)$.
    \item Remove \chris{the vector} $\mux(p_{ab}) = [1] \oplus (\mu(a) + \mu(b))$ from $\mux(\net)$.
    \item \chris{For each $\mux(v)$ in $\mux(\net)$, replace it with a vector $(\mux(v)_i)_{i\in S}$ where $S$ is the set $\{0\}\cup X\setminus \{b\}$.}
\end{enumerate}
Note, that because \pair{} is a cherry, the parent $p_{ab}$ of $a$ and $b$ will have $\mux$-vector $[1] \oplus (\mu(a) + \mu(b))$ before reduction and should be suppressed when reducing \pair{}. Note also that none of the in-degrees of any nodes have changed. Now let \pair{} be a simple reticulated cherry in $\mux(\net)$, we define the simple reticulated cherry reduction of \pair{} as the following operations:
\begin{enumerate}
    \item Remove \chris{the vector} $\mux(p_{a}) =[1] \oplus (\mu(a) + \mu(b))$ from $\mux(\net)$.
    \item Remove \chris{the vector} $\mux(p_b) = [2] \oplus \mu(b)$ from $\mux(\net)$.
    \item For each $\mux(v) \in \mux(\net) \setminus \mux(a)$ \chris{replace it with the vector $\mux(v) - \mux(v)_a  ([0] \oplus \mu(b))$.} 
\end{enumerate}
Note that tree-nodes have in-degree 1 and so, by \autoref{lem: ret cherry mu}, $\mux(p_a)$ is the $\mux$-vector of the parent of $a$, which should be suppressed when reducing \pair{}. Furthermore, because \pair{} is a simple reticulated cherry, the parent $p_b$ of $b$ has in-degree 2 before reducing and should be suppressed as well. Furthermore, because the edge between $p_a$ and $p_b$ there are no longer any paths to $b$ via $p_a$. Note that the in-degrees of any nodes that are not suppressed have not changed. Finally, we define the complex reticulated cherry reduction as follows:
\begin{enumerate}
    \item Remove \chris{the vector} $\mux(p_{a}) = [1] \oplus (\mu(a) + \mu(b))$ from $\mux(\net)$.
    \item Let $\mux(p_b) \in \mux(\net)$ be the vector with $\mux(p_b)_0 > 1$ and $\mux(p_b) = \mux(p_b)_0 \oplus \mu(b)$ and lower $\mux(p_b)_0$ by 1.
    \item For each $\mux(v) \in \mux(\net) \setminus \mux(a)$ \chris{replace it with the vector $\mux(v) - \mux(v)_a  ([0] \oplus \mu(b))$.} 
\end{enumerate}
For this reduction we keep the $\mux$-vector of the parent $p_b$ of $b$, because it is not suppressed when \pair{} is reduced in $\net$, because it has in-degree greater than 1 after reduction, but we do lower its in-degree by 1. Note that, by \autoref{lem: stack free leaf mults}, in stack-free networks there can only be one non-leaf node with $\mu$-vector equal to $\mu(b)$ and therefore $\mux(p_b)$ has multiplicity 1 in $\mux(\net)$. Finally note that the in-degrees of any other nodes, besides $p_b$ have not changed.

\begin{lemma} \label{lem: equivalence mu reduction}
    Let \pair{} be a reducible pair in $\mux(\net)$, the multiset generated by reducing \pair{} in $\mux(\net)$ is the $\mux$-representation of the network generated by reducing \pair{} in $\net$. 
\end{lemma}
The proof of this lemma is given in the appendix. Using these three lemmas we show the following isomorphism result for any two networks $\net_1$ and $\net_2$.
\begin{theorem} \label{thm: orchard encoding}
    Let $\net_1$ be semi-binary stack-free orchard and let $\net_2$ be semi-binary stack-free. Then,
    \begin{equation*}
        \mux(\net_1) = \mux(\net_2) \text{ if and only if } \net_1 \cong \net_2 
    \end{equation*}
\end{theorem}
Note that this means that semi-binary stack-free orchard networks are encoded by their modified $\mu$-representation in the class of semi-binary stack-free networks.
\begin{proof}
Suppose we are given a semi-binary stack-free orchard network $\net_1$, and a semi-binary stack-free network $\net_2$ with $\mux(\net_1) = \mux(\net_2)$. Note that $\mux(\net_1) = \mux(\net_2)$ implies that also  $\mu(\net_1) = \mu(\net_2)$. Then, because $\net_1$ is orchard it must contain a reducible pair of leaves \pair{}. If the pair \pair{} is a cherry, then by \autoref{lem: cherry mu} it must be a cherry in $\mu(\net_1)$. Therefore, it is also a cherry in $\mu(\net_2)$ and thus again by \autoref{lem: cherry mu} it is a cherry in $\net_2$. In this case, if $\net_1'$ is the network generated by reducing \pair{} in $\net_1$, and $\mux(\net_1')$ is the $\mu$-representation of this network, then by \autoref{lem: equivalence mu reduction}, $\mux'(\net_1) = \mux(\net_1')$, where  $\mux'(\net_1)$ is the multiset generated by reducing \pair{} in $\mux(\net_1)$. Thus, because there is only a single way of reducing a cherry in the $\mux$-representation, we have that $\mux'(\net_1) = \mux'(\net_2)$ and by \autoref{lem: equivalence mu reduction}, $\mux'(\net_2)$ is the $\mux$-representation $\mux(\net_2')$ of the network generated by reducing \pair{} in $\net_2$. To conclude, after reducing the cherry \pair{} in both $\net_1$ and $\net_2$, the two networks still have the same $\mux$-representation. \medskip

Alternatively, if \pair{} is a reticulated cherry in $\net_1$ then by \autoref{lem: ret cherry mu}, it is a reticulated cherry in $\mu(\net_1)$. Therefore, by the same argument as before, it is a reticulated cherry in $\net_2$. Furthermore, because $\mux(\net_1) = \mux(\net_2)$, if \pair{} is simple in $\net_1$ then it is simple in $\net_2$ and otherwise it is complex in both networks. As for each type of reticulated cherry there is a single way of reducing it in the $\mux$-representation, we again see that both networks must have the same modified $\mux$-representation after reduction of \pair{}.\medskip

Moreover, because $\net_1$ is orchard, the network will still be orchard after reducing the pair \pair{}. Therefore, it will again contain a reducible pair which is also a reducible pair in $\net_2$. It follows, that any sequence $S = s_1,s_2,\ldots,s_n$ of reducible pairs $s_i$, which reduces $\net_1$ (to a network on a single leaf), will also be a sequence of reducible pairs for $\net_2$. Furthermore, because $\net_1$ and $\net_2$ start out with the same set of leaves and each cherry reduction removes the same leaf from both networks, $S$ will also reduce $\net_2$ to a network on a single leaf, and it will be the same leaf. We will show that $\net_1$ and $\net_2$ are isomorphic by an inductive proof. Let $\net_1^{(i)}$ and $\net_2^{(i)}$ be the networks generated from $\net_1$ and $\net_2$ by performing reductions $s_1$ up to $s_i$, and let $\net_1^{(0)} = \net_1$ and $\net_2^{(0)} = \net_2$. And let us assume that the networks $\net_1^{(i)}$ and $\net_2^{(i)}$ are isomorphic. This is true for the base case where $i = n$, such that $\net_1$ and $\net_2$ are both reduced to a network on a single leaf by the entire sequence $s_1,s_2,\ldots,s_n$. Now take the networks $\net_1^{(i-1)}$ and $\net_2^{(i-1)}$ generated by performing reductions $s_1,s_2,\ldots,s_{i-1}$. By Corollary 1 of \cite{janssen2021cherry}, there is exactly one way to generate $\net_1^{(i-1)}$ and $\net_2^{(i-1)}$ from $\net_1^{(i)}$ and $\net_2^{(i)}$, respectively. 
From this it follows that, because $\net_1^{(i)}$ and $\net_2^{(i)}$ are isomorphic, we also have that $\net_1^{(i-1)}$ and $\net_2^{(i-1)}$ are isomorphic. Finally, because we have shown that the networks are isomorphic for $i = n$ and that they are isomorphic for $i = j-1$ if they are isomorphic for $i = j$, we can conclude that they are isomorphic for $i = 0$. This means $\net_1$ and $\net_2$ are isomorphic.
\end{proof}

\subsection{The \texorpdfstring{$\mux$}{m}-distance as a metric}\label{subsec:mux_metric}
By the definition as set out in \Cref{sec: preliminaries}
the symmetric difference between two multisets is empty, if, and only if, they are the same multiset. Furthermore, because the $\mu$-representation of a network is well-defined, if two networks are isomorphic then their $\mu$-representations are equal. If however, two networks have equal $\mu$-representation, this does not necessarily mean they are isomorphic, see the examples in \autoref{fig: counterYuki} and \autoref{fig: stack-free counter}. \autoref{thm: orchard encoding} shows that given two semi-binary stack-free networks with equal modified $\mu$-representations, if one of them is orchard, then they are isomorphic.
Let us define the \emph{$\mux$-distance} on networks~$\net_1$ and~$\net_2$ by taking the cardinality of the symmetric difference of modified $\mu$-representations, i.e., $d_{\mux}(\net_1, \net_2) = |\mux(\net_1) \triangle \mux(\net_2)|$.
By \autoref{thm: orchard encoding}, this is a metric on the class of semi-binary stack-free orchard networks.

\subsection{Non-binary stack-free orchard networks}
\begin{figure}
    \centering
    \includegraphics[width=0.9\textwidth]{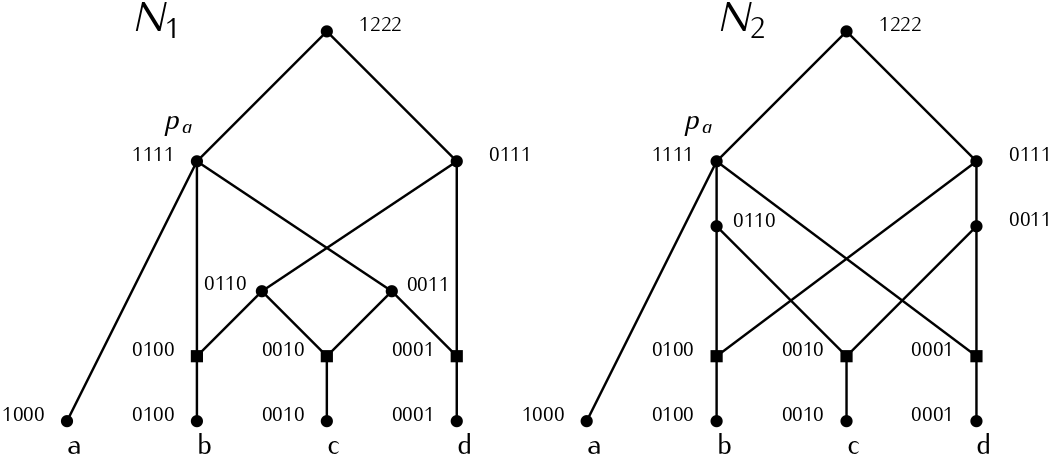}
    \caption{The networks $\net_1$ and $\net_2$ are both non-binary stack-free orchard with the same $\mux$-representation and equal out-degrees, however they are non-isomorphic. In $\net_1$, \pair{} is a reticulated cherry, while in $\net_2$ it is not. Similarly in $\net_2$, $(d,a)$ is a reticulated cherry, while in $\net_1$ it is not. This situation arises because the parent $p_a$ of leaf $a$ has out-degree 3. Notably $p_a$ is the only node with degree greater than 3.}
    \label{fig: Non-binary reticulated cherry counter example}
\end{figure}
In this section we will discuss whether our encoding results for the modified $\mu$-representation extend to non-binary orchard networks. \chris{We claim that} \autoref{lem: stack free leaf mults} regarding the parents of leaf nodes, the multiplicity of the $\mu$-vectors of leaf nodes, and 
the $\mu$-vector of the tree-node parent of a leaf node, holds for non-binary stack-free networks without any further modification. Similarly, we \chris{claim} \autoref{lem: cherry mu} holds for non-binary stack-free networks and therefore cherries are uniquely determined by the $\mu$-representation for non-binary stack-free networks.\medskip

However, \autoref{lem: ret cherry mu} does not have an obvious equivalent for non-binary networks. When considering whether the leaf pair \pair{} is a reticulated cherry, we can no longer require the existence of $\mu(p_a) = \mu(a) + \mu(b)$, because the parent of $a$ may have more children than just $a$ and $p_b$. By \autoref{lem: leaf tree parent} we can find the $\mu$-vector of the parent of $a$, but if it is not equal to $\mu(a) + \mu(b)$, then it is impossible to determine whether there are any other nodes on the path $p_a\pth p_b$. \autoref{fig: Non-binary reticulated cherry counter example} displays two non-binary stack-free orchard networks with the same $\mux$-representation, which are not isomorphic. Because the parent $p_a$ of $a$ has three children and its $\mu$-vector $\mu(p_a)$ is equal to the sums of the $\mu$-vectors of 2 different sets of 3 nodes ($1111=1000 + 0100 + 0011=1000 + 0110 + 0001$) of which all those that differ belong to reticulations, it is impossible to tell which set belongs to the children of $p_a$ and therefore whether \pair{} is a reticulated cherry or not. As a consequence of the \chris{counter}example given in \autoref{fig: Non-binary reticulated cherry counter example} we obtain the following:

\begin{theorem} \label{thm: non-binary orchard not encoded}
    Non-binary stack-free orchard networks are not encoded by their $\mux$-representation.
\end{theorem}

Note that nodes with the same $\mux$-vectors in $\net_1$ and $\net_2$ also have the same out-degrees. This means that the logical extension of the modified $\mu$-representation which adds the out-degrees of nodes does not lead to an encoding result for the class of non-binary stack-free orchard networks. 

\section{Encoding by \texorpdfstring{$\mu$}{m}-representations} \label{sec: strongly ret-vis encoding} 

In \cite{cardona2008comparison}, Cardona et al. showed that the $\mu$-representation serves as an encoding for non-binary tree-child phylogenetic networks. We have shown that for more general networks, even if they are stack-free semi-binary, two non-isomorphic networks may have the same $\mu$-representation. See \autoref{fig: counterYuki} for an example. \chris{We have also shown that for two semi-binary stack-free networks $\net, \net'$ with $\mux(\net) = \mux(\net')$ if either of them is orchard then they are isomorphic (\autoref{thm: orchard encoding}). Therefore, if we can find a subclass where $\mu(\net) = \mu(\net')$ implies that the nodes in $\net$ and $\net'$ have the same in-degrees} then we can show that equivalent $\mu$-representations imply isomorphism as long as one of the networks is orchard. This would give us an encoding result for this subclass.
\chris{In this section we set out to show the following main result.}

\begin{theorem*}[\ref{thm: strong ret vis encoding}]
    Let $\net_1$ and $\net_2$ be two semi-binary stack-free networks with $\mu(\net_1) = \mu(\net_2)$. Let $\net_1$ be strongly reticulation-visible and orchard. Then, $\net_1 \cong \net_2$.
\end{theorem*}

\chris{We want to show that there exists a subclass of phylogenetic networks, which we call \emph{strongly reticulation-visible}, that can be encoded by the $\mu$-representation without modification. To do so, we first give a set of equations which relate the in-degrees of the nodes in a phylogenetic network to the $\mu$-representation (\Cref{thm: indeg formula}). 
We then show that for reticulation-visible networks, the $\mu$-representation uniquely determines the in-degrees of all nodes, as long as we know which of the $\mu$-vectors belong to reticulations, and if we know, for each reticulation, the leaf for which it is \emph{stable} (\Cref{thm: indeg ret vis}). 
A vertex $v$ is \emph{stable}, if there exists a leaf~$a$ such that all paths from the root to~$a$ visit~$v$.}

\chris{
We then ask if this can also be done for orchard networks.
Because orchard networks do not contain tree-clones, it is easy to determine which $\mu$-vectors belong to reticulations in these networks, namely those which have multiplicity greater than 1. However, determining the stability of reticulations from the $\mu$-vectors is not obvious. 
In \Cref{sec: mu stabilitiy} we set out a way to do this.
It turns out this can be done if one more condition is added: each reticulation must be the lowest reticulation above some bridge. At the end of \Cref{sec: prelim lemmas strong ret vis}, we show that if we know which $\mu$-vector belongs to the head of a bridge (which we call a \emph{bridge-node}), then we can determine which $\mu$-vector belongs to the reticulation which is lowest above that bridge (\Cref{lem: lowest ret above bridge}). 
For these reticulations, it follows that we can identify the leaves with which they are stable.
We finally give a characterization in \Cref{thm: bridge-node mu-rep} for the $\mu$-vectors of bridge-nodes in semi-binary stack-free networks. By combining these results with \Cref{thm: indeg ret vis}, it follows that for strongly reticulation-visible networks, the in-degrees are uniquely determined by the $\mu$-representation. Therefore, there is a bijection between the $\mu$-representation and the modified $\mu$-representation for these networks. It then follows from \Cref{thm: orchard encoding} what we set out to prove:  strongly reticulation-visible orchard networks are encoded by their $\mu$-representation (\Cref{thm: strong ret vis encoding}).}

\subsection{Determining the in-degrees of reticulations} \label{sec: in-degree equation}
In this section we set out a relationship between the in-degrees of reticulations and the number of paths to the leaves below that reticulation. Naturally, higher in-degrees correspond to a larger number of paths crossing a reticulation. Therefore, it follows that the $\mu$-vectors of ancestors of a reticulation may contain information pertaining to the in-degree of that reticulation. 
\subsubsection{An equation relating the in-degrees of reticulations and \texorpdfstring{$\mu$}{m}-vectors}

Let $\net = (V,E)$ be \chris{an arbitrary} X-DAG. Recall that $\rho$ denotes the root of the network.
\begin{proposition} \label{thm: indeg formula}
    Let 
    $a$ be an element of X. Let $R$ be the set of reticulations in $V$. Then,
   \begin{align*}
        \mu(\rho)_a = \sum_{r_i \in R} (\indeg{r_i} - 1)\mu(r_i)_a + 1.
    \end{align*}
\end{proposition}
Note that $\mu(r_i)_a = 0$ if and only if $r_i$ is not an ancestor of~$a$, which means that the contributions to the sum of \autoref{thm: indeg formula} come from reticulation ancestors of~$a$. 
\begin{proof}
    We prove the theorem by induction on the hybridization number, $k = h(\net)$. We start by considering the base case, $k=0$.\medskip


Let us consider an $X$-DAG with $k = 0$, i.e., a tree. If the hybridization number is zero then the graph contains no reticulations and therefore $R$ is empty. This also means there is a unique tree-path from the root to each leaf 
and thus $\mu(\rho)_a = 1$. This shows the equation holds for each leaf of this graph.\medskip

 Now suppose the equation holds for each leaf of any $X$-DAG with hybridization number lower than $k$. Let $\net$ be an $X$-DAG with reticulation set $R$ such that $h(\net) = k$. Without loss of generality let $r_j$ be a highest reticulation in $\net$, which means $r_j$ has no reticulation ancestors. We can decrease the in-degree of $r_j$ by deleting an incoming edge $ur_j$. If $u$ is an elementary node we delete all edges on the maximal elementary path which visits $u$, as well as the nodes which become isolated by doing so. When we delete any of the incoming edges $ur_j$ in this way the resulting network $\net'$ is still an $X$-DAG. Furthermore, the in-degrees of any other reticulations in the network have not decreased, which means $\delta^-_{\net'}(r_i) = \delta^-_\net(r_i)$ for any $i \neq j$. Moreover, the number of paths from any reticulation $r$ to any leaf $a$ has not changed and thus $\mu'_a(r) = \mu(r)_a$. Finally, the hybridization number of $\net'$ is equal to $k-1$ and by the induction hypothesis we have,
 \begin{equation*}
     \mu'(\rho)_a = \sum_{r_i \in R'} (\delta^-_{\net'}(r_i) - 1)\mu(r_i)_a + 1.
 \end{equation*}
 Now there are two cases to consider:
 \begin{enumerate}
     \item $\delta^-_\net(r_j) = 2$
     \item $\delta^-_\net(r_j) > 2$.
 \end{enumerate}
 In the first case, the in-degree of $r_j$ after deleting an incoming edge becomes 1, which means $r_j \notin R'$. In this case $\delta^-_\net(r_j) - 2 = 0$. Which gives us:
\begin{align*}
    \mu'(\rho)_a    &= \sum_{r_i \in R'} (\delta^-_{\net'}(r_i) - 1)\mu(r_i)_a + 1 \\
                    &= \sum_{r_i \in R\setminus\{r_j\}} (\delta^-_\net(r_i) - 1)\mu(r_i)_a + (\delta^-_\net(r_j) - 2)\mu(r_j)_a + 1
\end{align*}
In the second case $r_j \in R'$ and, because we did not decrease the in-degree of any other reticulations, $R = R'$. The only difference is,  $\delta^-_{\net'}(r_j) = \delta^-_\net(r_j) - 1$. Therefore, we have:
\begin{align*}
    \mu'(\rho)_a    &= \sum_{r_i \in R'} (\delta^-_{\net'}(r_i) - 1)\mu(r_i)_a + 1 \\
                    &= \sum_{r_i \in R} (\delta^-_{\net'}(r_i) - 1)\mu(r_i)_a + 1 \\
                    &= \sum_{r_i \in R\setminus\{r_j\}} (\delta^-_\net(r_i) - 1)\mu(r_i)_a + (\delta^-_{\net'}(r_j) - 1)\mu(r_j)_a + 1\\
                    &= \sum_{r_i \in R\setminus\{r_j\}} (\delta^-_\net(r_i) - 1)\mu(r_i)_a + (\delta^-_\net(r_j) - 2)\mu(r_j)_a + 1.
\end{align*}
Note that after simplification the equation is the same for both cases.\medskip

When we add the edge $ur_j$ (or the maximal elementary path which visits $u$ and ends in $r_j$) back to the network, we generate $\net$ from $\net'$. In doing so we increase $\mu'_a(\rho)$ by $\mu(r_j)_a$ for any leaf $a$. 
Indeed, $r_j$ is chosen to be a highest reticulation in $\net$, and so there is a single path from the root to $r_j$ which uses the edge $(u,r_j)$.
By definition, there are $\mu(r_j)_a$ paths from $r_j$ to any leaf $a$. This means that, for any leaf $a$
\begin{align*}
    \mu(\rho)_a     &= \mu'(\rho)_a + \mu(r_j)_a \\
                    &= \sum_{r_i \in R\setminus\{r_j\}} (\indeg{r_i}_{\net} - 1)\mu(r_i)_a + (\indeg{r_j}_\net - 2)\mu(r_j)_a + 1 + \mu(r_j)_a \\
                    &= \sum_{r_i \in R\setminus\{r_j\}} (\indeg{r_i}_{\net} - 1)\mu(r_i)_a + (\indeg{r_j}_\net - 1)\mu(r_j)_a + 1\\
                    &= \sum_{r_i \in R} (\indeg{r_i}_{\net} - 1)\mu(r_i)_a + 1,
\end{align*}
which shows that the equation in the theorem holds for any leaf of $\net$. Thus, we have shown that if the equation holds for any leaf of an $X$-DAG with hybridization number $k-1$, then it holds for any leaf of an $X$-DAG with hybridization number $k$. We had already shown the equation holds for any leaf for the base case $k=0$. Therefore, we can conclude the equation holds for any leaf of any $X$-DAG, which proves the theorem.
\end{proof}

Note that this theorem holds for non-binary $X$-DAGs, as we never assumed a limit on the in-degree of reticulations or the out-degree of tree-nodes. 
\autoref{thm: indeg formula} provides us with a system of linear equations that govern the in-degrees of reticulations as a function of $\mu$-vectors. We know this system of equations must have a solution as long as it belongs to a valid phylogenetic network. However, there are no guarantees yet that it has a unique solution. A network may contain more reticulations than there are linearly independent equations in the system. The network displayed in \autoref{fig: indeg example} is semi-binary stack-free orchard, yet the system of equations generated by applying \autoref{thm: indeg formula} to each leaf does not have a unique solution. Note that leaf $c$ has no reticulation ancestors and therefore the equation for leaf $c$ ($\mu(\rho)_c=1$) does not contribute. This means that there is one equation for leaf $a$ and one for leaf $b$, but three variables, the in-degrees of reticulations $r_1,r_2$ and $r_3$. Note that the lower bound for the in-degree of reticulations is 2, and therefore, there are only a finite number of solutions. In this case there are two solutions to the system of equations and \autoref{fig: indeg example} shows $(\indeg{r_1},\indeg{r_2},\indeg{r_3}) =(3,3,2)$ is the one belonging to this network.

\begin{figure}
    \centering
    \includegraphics[width = 0.9\textwidth]{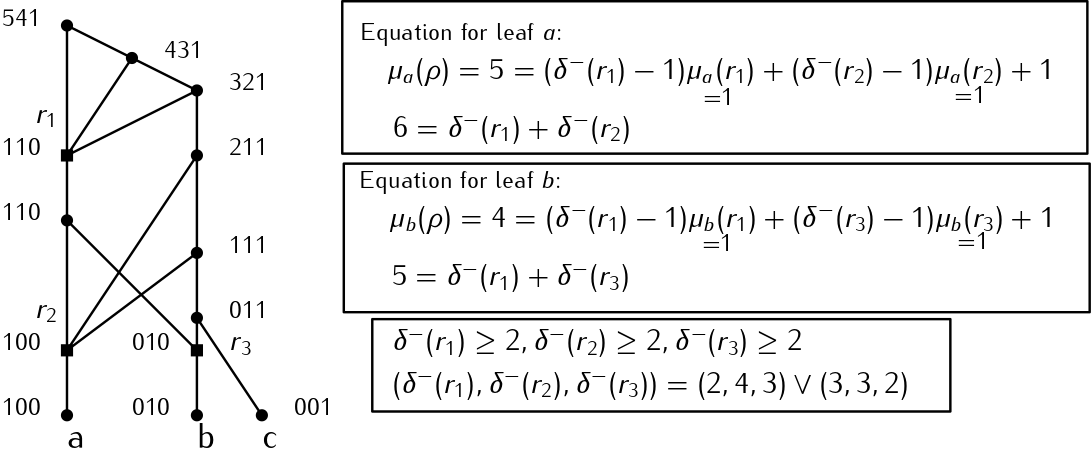}
    \caption{A semi-binary stack-free orchard network for which the system of equations generated by applying \autoref{thm: indeg formula} to each leaf does not have a unique solution. However, the solution $(\indeg{r_1},\indeg{r_2},\indeg{r_3}) =(3,3,2)$ is the unique one that belongs to this network.}
    \label{fig: indeg example}
\end{figure}

\subsubsection{In-degrees of stable reticulations} \label{sec: ret-vis in-degrees}
In the last section we presented a theorem which allows us to derive a system of equations which govern the in-degrees of reticulations as a function of the $\mu$-vectors. This system of equations applies to every network with the same $\mu$-representation and set of reticulations $R$. However, the system of equations does not always have a unique solution. In this section we will introduce specific conditions under which it does have a unique solution. We can extend the results of the previous section to the class of networks called reticulation-visible. Recall the definition of a reticulation-visible network:\medskip

\begin{definition}
    An X-DAG with reticulation set $R$ is called reticulation-visible if each reticulation $r \in R$ is stable, that is there exists a leaf $a \in X$ such that all paths from the root to $a$ visit $r$.
\end{definition}
If we let $\net$ be a reticulation-visible X-DAG with reticulation set $R$ then the following is true. For $r \in R$, with leaf $a$ below $r$ such that all paths from the root to $a$ visit $r$, and $v$ any ancestor of $r$,
\begin{equation} \label{eq: stable ret paths}
    \mu(v)_a = P_{vr}\mu(r)_a,
\end{equation}
where $P_{vr}$ is the number of paths from $v$ to $r$. This is true, as every path from $v$ to $a$ must be the concatenation of a path from $v$ to $r$ and a path from $r$ to $a$. With this we gain the following lemma. \chris{This shows that in reticulation-visible networks, there is a unique equation for each reticulation.}

\begin{lemma} \label{lem: indeg formula retvis}
    Let $\mathcal{N}$ be an X-DAG with reticulation set $R$. Let $r_\ell$ be any reticulation in $R$ and let $A$ be the set of ancestors of $r_\ell$. Suppose there exists a leaf $a \in X$  such that all paths from the root to $a$ visit $r_\ell$. Then,
   \begin{align} \label{eq: indeg ret vis}
        \frac{\mu(\rho)_a}{\mu(r_\ell)_a}
                    &= \sum_{r\in A}(\indeg{r} - 1) \frac{\mu(r)_a}{\mu(r_\ell)_a} + 1
    \end{align}
\end{lemma}
\begin{proof}
    Let us generate $\net'$ from $\net$ by first attaching a new leaf $a'$ to $r_\ell$, by adding the edge $(r_\ell,a')$, and then adjusting the $\mu$-representation by adding a column for $a'$, such that $\mu'(v) = \mu(v) \oplus \mu(v)_{a'}$. Note that this network is now an $X'$-DAG, where $X' = X\cup \{a'\}$. Then by \autoref{thm: indeg formula} we have:
    \begin{equation} \label{eq: in-degree net'}
        \mu'_{a'}(\rho) = \sum_{r \in R}(\indeg{r} - 1)\mu'_{a'}(r) + 1.
    \end{equation}
    Now note that $r_\ell$ is the lowest reticulation above $a'$ by construction and therefore the only reticulations that contribute to the sum in \autoref{eq: in-degree net'} are the reticulations in $A$.
    Furthermore, the number of paths from any ancestor $v$ of $a'$, which is not $a'$, to $a'$ are equal to the number of paths from $v$ to $r_\ell$, which is the same in $\net$ and $\net'$, which means by \autoref{eq: stable ret paths},
    \begin{equation} \label{eq: paths net'}
        \mu'_{a'}(v) = P_{vr_\ell} = \frac{\mu(v)_a}{\mu(r_\ell)_a}. 
    \end{equation}
    Because $\rho$ is an ancestor of $a'$, as are the elements of $A$, we can substitute \autoref{eq: paths net'} in \autoref{eq: in-degree net'} to get:
    \begin{equation*}
        \frac{\mu(\rho)_a}{\mu(r_\ell)_a}
                        = \sum_{r\in A}(\indeg{r} - 1) \frac{\mu(r)_a}{\mu(r_\ell)_a} + 1.
    \end{equation*}
\end{proof}

For the following lemmas, let $\net_1$ and $\net_2$ be two reticulation-visible X-DAG's \chris{on the same leaf set~$X$} with $\mu(\net_1) = \mu(\net_2)$. 
\chris{Let~$R_1$ and $R_2$ denote the reticulation sets of~$\net_1$ and~$\net_2$, respectively.
We say that~$\net_1$ and~$\net_2$ \emph{agree on the reticulation set} if there exists a bijection~$f:R_1\rightarrow R_2$ such that~$\mu(r_1) =\mu(f(r_1))$ for all~$r_1\in R_1$.
For simplicity, we refer to~$R_1$ (and by bijection, also~$R_2$) as~$R$.
We assume that~$\net_1$ and~$\net_2$ agree on the reticulation set.}
Finally, assume that for each reticulation $r \in R$, there exists a leaf $a \in X$ such that $r$ is stable with respect to $a$ in both networks.

\begin{lemma} \label{lem: stable comparable}
    Each reticulation $r \in R$ has the same reticulation descendants and the same reticulation ancestors in $\net_1$ and $\net_2$.
\end{lemma}
\begin{proof}
    Let $r \in R$ be stable for leaf $a$ in both networks. Then all paths from the root to $a$ visit $r$ in both networks. Therefore, all reticulations $r_i \in R$, with $\mu(r_i) \geq \mu(a)$ must be on a path from the root to $a$ which also visits $r$ in both networks. Therefore, each $r_i \in R$, with $\mu(r_i) \geq \mu(a)$ must be either a descendant or an ancestor of $r$. If $\mu(r_i) \geq \mu(r)$ then $r_i$ must be an ancestor of~$r$ in both networks, and if $\mu(r_i) \leq \mu(r)$ then~$r_i$ must be a descendant of~$r$ in both networks. 
\end{proof}

\begin{lemma} \label{lem: in-degrees highest rets}
    Let $H \subseteq R$ be the subset of $R$ which contains only reticulations without other reticulation ancestors in either network. Then the reticulations in $H$ have the same in-degrees in both networks.
\end{lemma}
\begin{proof}
    For any reticulation $r \in H$ the set $A$ in \autoref{lem: indeg formula retvis} contains only $r$ and therefore the in-degree of $r$ is given directly by \autoref{eq: indeg ret vis}, applied to the leaf for which the reticulation is stable in both networks.
\end{proof}

By a similar reasoning we obtain the following:
\begin{lemma} \label{lem: indeg ret vis ancestors}
    Let $r_\ell$ in $R$ and let the in-degrees of the other ancestors of $r_\ell$ be the same in $\net_1$ and $\net_2$. Then the in-degree of $r_\ell$ is the same in $\net_1$ and $\net_2$.
\end{lemma}
\begin{proof}
   Given a reticulation $r_\ell \in R$, the set $A$ in \autoref{eq: indeg ret vis} contains only $r_\ell$ and its other ancestors, which are the same in both networks. Now assume the in-degrees of the ancestors of $r_\ell$ are fixed, then the in-degree of $r_\ell$ is given by \autoref{eq: indeg ret vis}, applied to the leaf for which the reticulation is stable in both networks.
\end{proof}

This leads to the following conclusion:
\begin{proposition} \label{thm: indeg ret vis}
    Let $\net_1$ and $\net_2$ be two reticulation-visible X-DAG's with $\mu(\net_1) = \mu(\net_2)$. Let both networks have the same reticulation set $R$ and the same leaf set $X$. Finally, assume that for each reticulation $r \in R$, there exists a leaf $a \in X$ such that $r$ is stable with respect to $a$ in both networks. Then the in-degrees of the reticulations are the same in both networks.
\end{proposition}
\begin{proof}
    Let $H$ be the set of highest reticulations in $R$, which is the same set for $\net_1$ and $\net_2$, by \autoref{lem: stable comparable}. By \autoref{lem: in-degrees highest rets} the in-degrees of the reticulations in $H$ are the same in both networks. Let $S$ be the set of second highest reticulations in $R$, such that each reticulation $r \in S$ only has ancestors in $H$. By applying \autoref{lem: indeg ret vis ancestors} the in-degrees of the reticulations in $S$ are also the same in both networks. Now, in the same way, the in-degrees of the third highest reticulations are equal and so on and so forth. This process must terminate as we only consider finite graphs. Thus, by repeated application of \autoref{lem: indeg ret vis ancestors}, we see that the in-degrees of all reticulations in $R$ are the same in both networks.
\end{proof}

\subsection{Determining stable nodes} 
\label{sec: mu stabilitiy}
In this subsection, we work with stack-free phylogenetic networks, denoted by $\net = (V,E)$. 
\autoref{thm: indeg ret vis} shows that if two reticulation-visible networks have equal $\mu$-representations, identical reticulation sets, and every reticulation is stable for at least one common leaf in both networks, then the in-degrees of the reticulations are guaranteed to be equal in both networks. These conditions ensure that the same system of equations that govern the in-degrees of reticulations holds for both networks and it has a unique solution. \medskip

To reiterate, we're trying to find the sub-class of semi-binary stack-free networks for which equal $\mu$-representations imply equal in-degrees. Now we have shown that this could be a sub-class of reticulation-visible networks, but  
we need to determine under which conditions the reticulation sets are equal and the reticulations are stable for common leaves. 
In the case of stack-free orchard networks the reticulations correspond to $\mu$-vectors with multiplicity 2. However, the stability of reticulations with respect to specific leaves may still differ between the networks, even if they share the same $\mu$-representation. Therefore, the next step will be to demonstrate the conditions under which the stability of a node with respect to a leaf is determined by the $\mu$-representation.

\subsubsection{Preliminary lemmas} \label{sec: prelim lemmas strong ret vis}
Recall that a tree-clone is a tree-node for which there exists another tree-node with the same $\mu$-vector.

\begin{lemma} \label{lem: tree-clones not stable}
    Tree-clones are not stable.
\end{lemma}
\begin{proof}
    Let $u,v$ be distinct tree-clones. Recall that for tree-nodes we have $\mu(v_1) > \mu(v_2)$ whenever $v_1 \neq v_2$ and $v_1$ is an ancestor of $v_2$. Therefore, because $\mu(u) = \mu(v)$, it is clear that $u,v$ are neither ancestors nor descendants of each other. This means a path never visits both $u$ and $v$. Furthermore, because the $\mu$-vectors are equal, we know that for each leaf, such that $u$ is on a path to that leaf, $v$ is also on a path to that leaf. By our previous statement these paths must be distinct and neither contains both nodes. Therefore, there are no leaves such that all paths to that leaf contain $u$, nor are there any leaves such that all paths to that leaf contain $v$.
\end{proof}

\begin{lemma} \label{lem: ret child stability}
    A reticulation is stable with respect to a leaf if, and only if, its child is stable with respect to that leaf.
\end{lemma}
\begin{proof}
    In stack-free phylogenetic networks, all paths from the root to a leaf which visit a reticulation must also visit its child and vice versa, as by definition there is a single edge leaving each reticulation and a single edge going into its child.
\end{proof}

\begin{corollary}
    Nodes with the same $\mu$-vector as tree-clones are not stable.
\end{corollary}

Either they are themselves tree-clones, or they are the reticulation parent of a tree-clone, which means by \autoref{lem: ret child stability} that they are not stable.

\begin{lemma} \label{lem: no tree-clones mult limit}
    Let $\mu(\net)$ be the $\mu$-representation of a stack-free network and let $\mu(v) \in \mu(\net)$ be a $\mu$-vector, such that there does not exist a pair of tree-clones whose $\mu$-vectors are equal to $\mu(v)$. Then $\mu(v)$ has multiplicity at most 2 in $\mu(\net)$.
\end{lemma}
\begin{proof}
    Assume there are no tree-clones with $\mu$-vectors equal to $\mu(v)$, but assume $\mu(v)$ has multiplicity more than 2. Then there must be at least 2 distinct reticulations $r_1,r_2 \in R$ with the same $\mu$-vector $\mu(r_1) = \mu(r_2) = \mu(v)$. As reticulations have the same $\mu$-vectors as their children and in a stack-free network the children of reticulations are tree-nodes, this means there exist two tree-nodes with $\mu$-vectors equal to $\mu(v)$. However, as we assumed there are no tree-clones with $\mu$-vectors equal to $\mu(v)$, we have reached a contradiction.
\end{proof}

\begin{corollary} \label{cor: stable mult limit}
    \chris{Let $\net$} be a stack-free network with $\mu$-representation $\mu(\net)$. If $\mu(v)$ has a multiplicity greater than 2 in $\mu(\net)$ then all nodes $v \in V$ with $\mu$-vector equal to $\mu(v)$ are not stable.
\end{corollary}

Instead of showing directly when the $\mu$-representation determines whether a $\mu$-vector belongs to a stable node, we will first show some other results which we will use for our argument. A \emph{bridge} or a \emph{cut-edge} is an edge, for which it holds that if the edge would be deleted the number of connected components of the graph goes up. In the case of phylogenetic networks, which are connected graphs, deleting a bridge makes the graph no longer connected.


Let the head of a bridge be called a \emph{bridge-node}. Then, \chris{since reticulation edges are never bridges}, bridge-nodes are tree-nodes. All leaves are automatically bridge-nodes, as they become isolated whenever the edge directed into them is deleted. The root is not a bridge-node because it has zero in-degree.

\begin{observation} \label{obs: bridges are stable_1}
    Let $v_b$ be a tree-node other than the root. Then, $v_b$ is a bridge-node if, and only if, all paths from ancestors of $v_b$ to leaves below $v_b$ pass through $v_b$.
\end{observation}
In other words, this \chris{observation} implies that a bridge-node is stable for all leaves below it. This means that if $v_b$ is a bridge-node, then nodes which lie on a path to a leaf below $v_b$ are either ancestors or descendants of $v_b$.



\begin{observation} \label{obs: bridges are stable_2}
    Let $v_b$ be a bridge-node. Then, for $u \in V$, $u \neq v_b$ 
    \begin{itemize}
        \item $\mu(u) \geq \mu(v_b) \iff \mu(u)$ belongs 
        to \chris{an} ancestor of $v_b$.
        \item $\mu(u) < \mu(v_b) \iff \mu(u)$ belongs
        to \chris{a} descendant of $v_b$.
    \end{itemize}
\end{observation}
Previously, these implications only held in one direction (right to left) for all tree-nodes.

\begin{observation} \label{obs: strongly stable}
    If there is a tree-path from the child $c_r$ of a reticulation $r$ to a bridge-node $v_b$, then $r$ is stable with respect to all leaves below $v_b$.
\end{observation}


Let us define the set $A_b$ as the subset of the ancestors of a bridge-node $v_b$, whose $\mu$-vectors have multiplicity \chris{exactly} 2 in $\mu(\net)$.
\chris{While we have defined $\mu$ on networks and on individual nodes, we generalize this definition to a set of vertices.}
For any set $S \subseteq V$, let $\mu(\net)_S$ be the multiset of $\mu$-vectors of the nodes in $S$.
\chris{In other words, $\mu(\net)_S = \{\mu(v)\in \mu(\net): v\in S\}$.}
Then $\mu(\net)_{A_b}$ is exactly the multiset $\{\mu(v): \mu(v) \geq \mu(v_b), \#\mu(v) = 2\}$.
By \autoref{obs: bridges are stable_2}, each vector in $\mu(\net)_{A_b}$ only belongs to ancestors of $v_b$. Notice that $\mu(\net)_{A_b}$ is purely defined in terms of $\mu$-vectors if given a $\mu(v_b)$ which belongs to a bridge-node, and it can be determined from the $\mu$-representation, without any knowledge of the graph.
\chris{Here, we show that it can be used to find $\mu$-vectors which belong to stable reticulations.
In \Cref{sec: strongly ret vis}, we show how to determine $\mu$-vectors that belong to bridge-nodes.}

\begin{lemma} \label{lem: tree-clones not minimal ancestors of bridge}
    If $\mu(v)$ is minimal in $\mu(\net)_{A_b}$, then it belongs to a reticulation and to the child of that reticulation.
\end{lemma}
\begin{proof}
    A $\mu$-vector with multiplicity 2 either belongs to a pair of tree-clones or to a reticulation and its child. Now, if $\mu(v_1) \in \mu(\net)_{A_b}$ belongs to a pair of tree-clones $v_1,v_2$ with $\mu(v_1) = \mu(v_2)$, then there cannot be a tree-path from either of them to a bridge-node $v_b$, because by \autoref{lem: tree-clones not stable}, tree-clones are not stable. Therefore, $v_1$ and $v_2$ must have reticulation descendants $r_1,r_2$, who are  ancestors of $v_b$. Note that $r_1$ could be equal to $r_2$. If the multiplicity of either $r_1$ or $r_2$ is greater than 2, then their child must be a tree-clone, by the contrapositive of \autoref{lem: no tree-clones mult limit}. By the same argument there must then be other reticulation descendants of $v_1$ and $v_2$ above $v_b$. We only consider finite graphs, and therefore, w.l.o.g. we can assume $\mu(r_1)$ and $\mu(r_2)$ have multiplicity 2 in $\mu(\net)$. Then, $r_1,r_2 \in A_b$, with $\mu(r_1)< \mu(v_1)$ and $\mu(r_2) < \mu(v_2)$. Therefore, $\mu(v_1)$ is not minimal in $A_b$. We can conclude that if $\mu(r)$ is minimal in $A_b$, then $r$ is a reticulation.
\end{proof}

\begin{lemma} \label{lem: lowest ret above bridge}
    There is a tree-path from the child of a reticulation $r$ to a bridge-node $v_b$ if, and only if, $\mu(r)$ is minimal in the set $\mu(\net)_{A_b}$.
\end{lemma}
\begin{proof}
    For the first direction, assume there is a tree-path from the child of a reticulation $r$ to a bridge $v_b$. Then, by \autoref{obs: strongly stable}, $r$ is stable. Therefore, by \autoref{cor: stable mult limit}, $\mu(r)$ has multiplicity 2 in $\mu(\net)$. Furthermore, $r$ is an ancestor of $v_b$, therefore $r \in A_b$. Also note, that 
    $r$ is the lowest reticulation above $v_b$. Then, by 
    \autoref{lem: tree-clones not minimal ancestors of bridge}, all ancestors of $v_b$, with $\mu$-vectors which have multiplicity 2, are ancestors of $r$. Thus, we have $\mu(v) \geq \mu(r)$, for $\mu(v) \in \mu(\net)_{A_b}$. Therefore, $\mu(r)$ is minimal in $A_b$. This proves the first direction.\medskip

    \chris{The other direction we will prove by contradiction.} Assume $\mu(r)$ is minimal in $\mu(\net)_{A_b}$. By \autoref{lem: tree-clones not minimal ancestors of bridge}, $\mu(r)$ belongs to a reticulation $r$ and its child $c_r$. Furthermore, because $\mu(c_r) = \mu(r) \in \mu(\net)_{A_b}$, both $r$ and $c_r$ are ancestors of $v_b$. Therefore, there is a path from $c_r$ to $v_b$. If the path from $c_r$ to $v_b$ is not a tree-path, then $r$ is not the lowest reticulation above $v_b$. Let $r_\ell$ be the lowest reticulation above $v_b$, then 
    $r_\ell$ is a descendant of $r$. By \autoref{obs: strongly stable}, $r_\ell$ is stable. Therefore, by \autoref{cor: stable mult limit}, $\#\mu(r_l) \leq 2$. The child of $r_\ell$ has the same $\mu$-vector, therefore $\#\mu(r_l) = 2$.
    This means that $\mu(r_\ell) \in \mu(\net)_{A_b}$ and because $r_\ell$ is a descendant of $r$, also $\mu(r)>\mu(r_\ell)$. This contradicts our assumption that $\mu(r)$ is minimal in $\mu(\net)_{A_b}$. Therefore, $r$ is the lowest reticulation above $v_b$, and the path from $c_r$ to $v_b$ is a tree-path.
\end{proof}

\chris{\Cref{lem: lowest ret above bridge} shows that given a vector $\mu(v_b)$ which belongs to a bridge-node, we can determine the $\mu$-vector belonging to the lowest reticulation above that bridge by finding the minimal $\mu$-vector $\mu(v) \geq \mu(v_b)$ with multiplicity 2.}

\subsubsection{Strongly reticulation-visible networks} \label{sec: strongly ret vis}

\begin{figure}
    \centering
    \includegraphics[width = 0.2\textwidth]{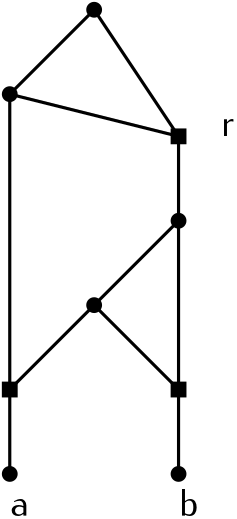}
    \caption{A phylogenetic network which is reticulation-visible but not strongly reticulation-visible. Reticulation $r$ is stable with respect to leaf $b$, but there is no tree-path from the child of $r$ to a bridge. The only bridges are the edges ending in $a$ and $b$.}
    \label{fig: not strongly retvis}
\end{figure}

The lemmas in the previous section show that if we know that a given $\mu$-vector belongs to a bridge-node, then we can find the $\mu$-vector of the lowest reticulation above this bridge-node and the reticulation will be stable for all leaves below the bridge-node (which are the leaves for which the $\mu$-vector of the bridge-node has non-zero coordinates). We therefore propose to consider the class of networks such that for each reticulation there is a tree-path from its child to a bridge. We will call this the class of \emph{strongly reticulation-visible networks}.
\medskip

Note that all strongly reticulation-visible networks are reticulation-visible. But there are reticulation-visible networks which are not strongly reticulation-visible. See \autoref{fig: not strongly retvis} for an example. Therefore the class of strongly reticulation-visible networks is a proper subclass of the class of reticulation-visible networks, which is itself a proper subclass of stack-free networks. If we can show that for strongly reticulation-visible networks it is possible to determine whether a $\mu$-vector belongs to a bridge-node from just the $\mu$-representation (\autoref{thm: bridge-node mu-rep}), we will have shown that the stability of reticulations can be determined. 
Knowing the stability of reticulations will show that the reticulation in-degrees can be determined (\autoref{thm: indeg ret vis}), finally leading to our main isomorphism result (\autoref{thm: strong ret vis encoding}). \medskip


We quickly summarize the collection of results obtained so far.
By \autoref{obs: bridges are stable_1}, bridge-nodes are stable. Therefore, if $v_b$ is a bridge-node, then by \autoref{cor: stable mult limit}, $\#\mu(v_b) \leq 2$. Furthermore, by \autoref{lem: tree-clones not stable}, $\mu(v_b)$ does not belong to a pair of tree-clones. Finally, all leaves are bridge-nodes, therefore all unit vectors in $\mu(\net)$ belong to bridge-nodes, so we only still have to consider $\mu$-vectors which are not unit vectors. \chris{Note that we also do not need to consider the unique maximal $\mu$-vector in $\mu(\net)$.}
\begin{observation}\label{obs:root mu vector}
    The $\mu$-vector~$\mu(\rho)$ belonging to the root of a network~$\net$ is the unique maximal $\mu$-vector in~$\mu(\net)$. 
\end{observation}
\chris{Now let $\net = (V,E)$ be a semi-binary stack-free network. \Cref{thm: bridge-node mu-rep} gives a set of conditions on the $\mu$-representation $\mu(\net)$ which determine whether a given $\mu$-vector belongs to a bridge-node.}

\begin{proposition} \label{thm: bridge-node mu-rep} 
    Let~$\mu(v_b)$ be a non-unit vector that is not $\mu(\rho)$.
    Then, it belongs to a bridge-node if, and only if, 
    \begin{enumerate}
        \item there is exactly one pair $\mu(k),\mu(\ell) \in \mu(\net)$ such that $\mu(v_b) = \mu(k) + \mu(\ell)$, and
        \item \chris{one of the following holds}
        \begin{enumerate}
            \item \chris{at least one of} $\mu(k)$ and $\mu(\ell)$ \chris{has multiplicity 1} in $\mu(\net)$; or
            \item $\mu(k) = \mu(\ell)$ \chris{and} $\#\mu(k) = 2$, and
        \end{enumerate}
        \item $\mu(\net)$ does not contain vectors $\mu(x),\mu(y),\mu(z)$, such that $\mu(z) \nleq \mu(v_b)$, $\mu(z) = \mu(x) + \mu(y)$ and $\mu(x) < \mu(v_b)$.
    \end{enumerate}
\end{proposition}
\begin{proof}
\chris{For a more formal proof see the appendix. Here we would like to give a proof by illustration. First we will show that the three conditions hold for any bridge node $v_b$.}

\begin{figure}[ht]
    \centering
    \includegraphics[width=0.5\linewidth]{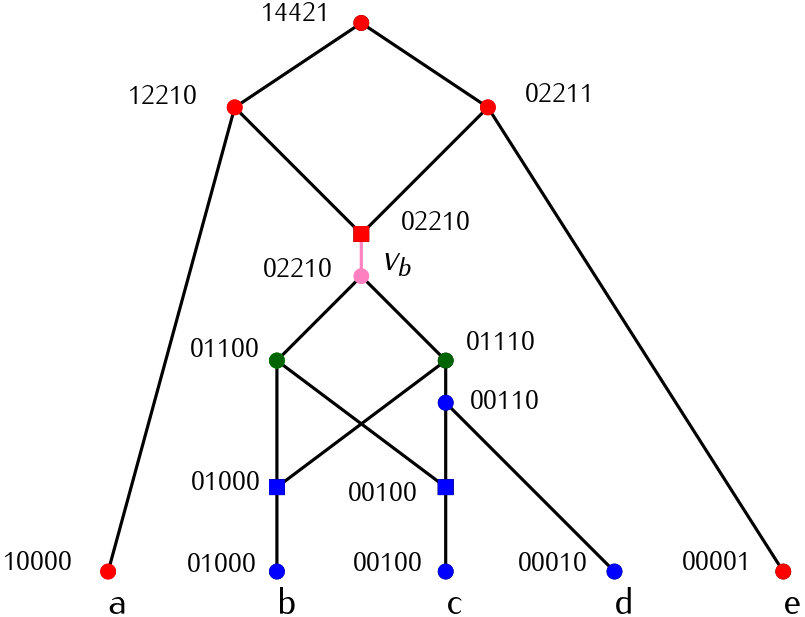}
    \caption{A semi-binary stack-free network with bridge-node $v_b$.}
    \label{fig: no other sum for bridge node children}
\end{figure}

 \chris{Observe \Cref{fig: no other sum for bridge node children} for an example of condition 1. Note that none of the red `outside' nodes have $\mu$-vectors which could ever contribute to a pair whose $\mu$-vectors sum up to $\mu(v_b)$ because their $\mu$-vectors are simply too large and/or contain non-zero coordinates that correspond to leaves that are not below $v_b$. Note also that none of the $\mu$-vectors of the blue `inside' nodes (which are not children of $v_b$) could ever contribute to a sum of only two $\mu$-vectors which sum up to $\mu(v_b)$ because they are simply too small. Only the green child nodes of $v_b$ have $\mu$-vectors which together sum up to $\mu(v_b)$. There is always just a single pair of $\mu$-vectors which sum up to $\mu(v_b)$ if $v_b$ is a bridge-node and they belong to the children of $v_b$.}\medskip

 \chris{For the second condition consider what it means if either \textit{(a)} or \textit{(b)} would not hold. Then the children of $v_b$ have distinct $\mu$-vectors which both have multiplicity 2 or higher or they have equal $\mu$-vectors and the multiplicity is 3 or higher. There are three such possible cases: either both children are reticulations (\autoref{fig: not two reticulation children}), or one of the children is a reticulation and the other is a tree-clone (\autoref{fig: not a ret and tree-clone child}), or both are tree-clones (\autoref{fig: not two tree-clone children}). In the final case, they either have distinct $\mu$-vectors, or they have equal $\mu$-vectors and there is a third node with the same $\mu$-vector.}

 \begin{figure}[ht]
     \centering
     \includegraphics[width=0.4\linewidth]{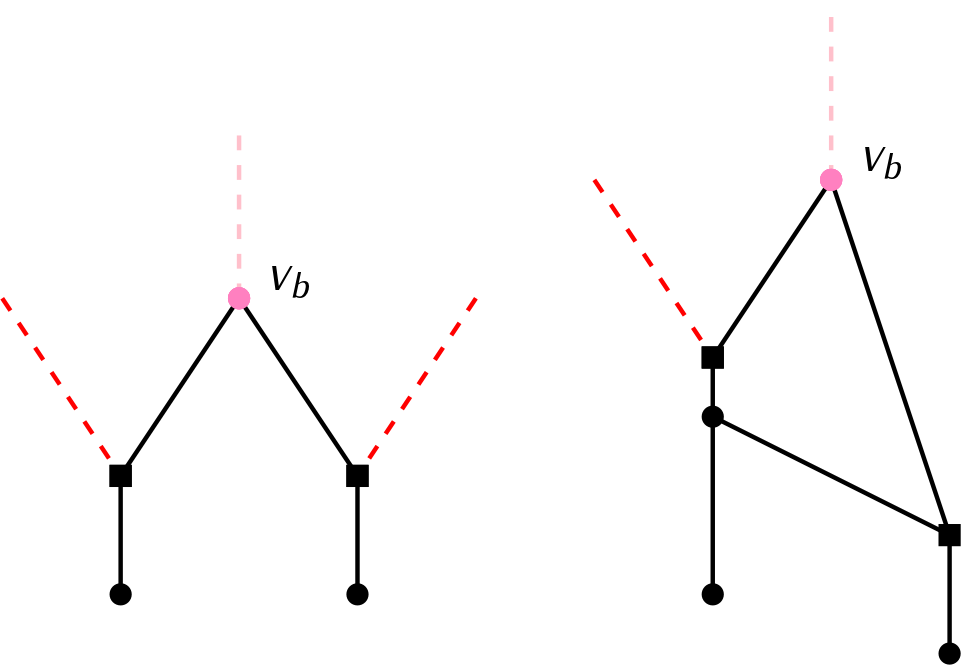}
     \caption{Two networks in which $v_b$ has two reticulation children. The dashed lines indicate that they connect somewhere above.}
     \label{fig: not two reticulation children}
 \end{figure}

 \chris{\Cref{fig: not two reticulation children} shows that even if one of the reticulation children of $v_b$ is a descendant of the other, there is a reticulation edge going into one of the reticulations (the red dashed line) which is not coming from $v_b$ and therefore $v_b$ is not a bridge-node.}

 \begin{figure}[ht]
     \centering
     \includegraphics[width=0.6\linewidth]{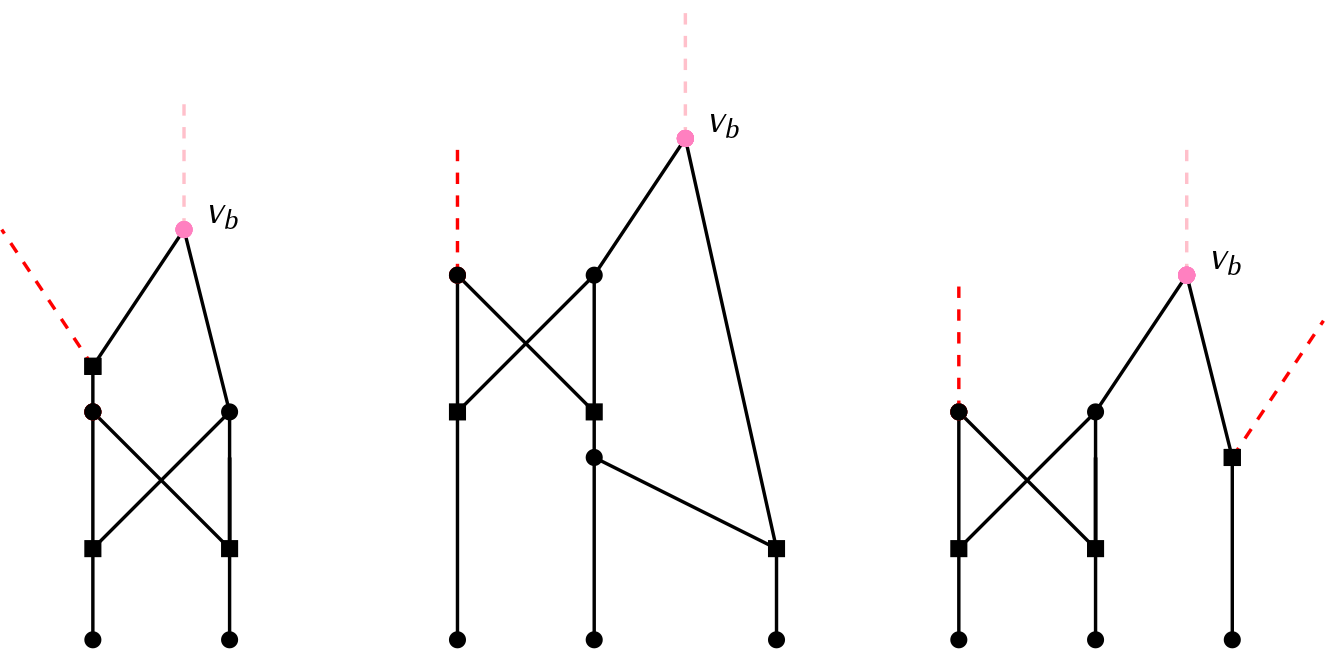}
     \caption{Three networks in which $v_b$ has a reticulation and a tree-clone child.}
     \label{fig: not a ret and tree-clone child}
 \end{figure}

\chris{\Cref{fig: not a ret and tree-clone child} illustrates the case when $v_b$ has a reticulation and a tree-clone as children. Note that even if either child is a descendant of the other child, then there is still a path (via the red dashed line) to some of the descendants of $v_b$ which does not pass through $v_b$. Therefore, $v_b$ cannot be a bridge-node.}

\begin{figure}[ht]
    \centering
    \includegraphics[width=0.5\linewidth]{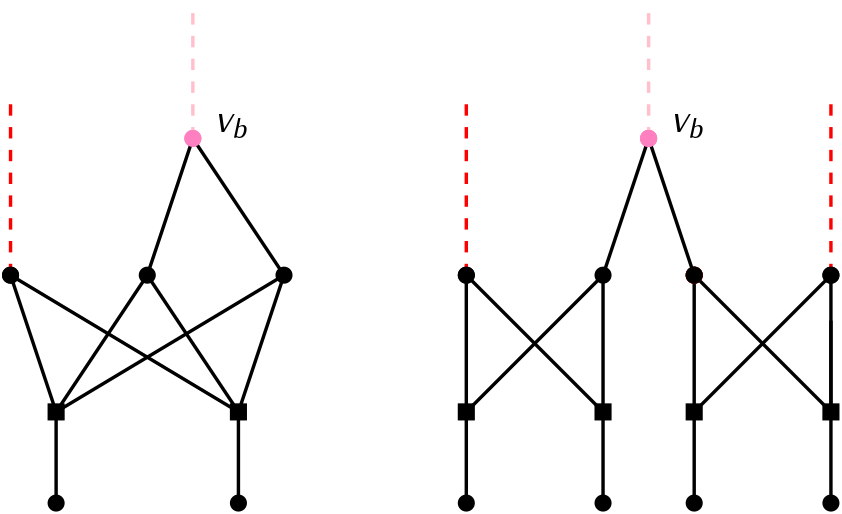}
    \caption{Two networks in which $v_b$ has two tree-clone children.}
    \label{fig: not two tree-clone children}
\end{figure}

\chris{Now assume $v_b$ has two children which are tree-clones which either have the same $\mu$-vector with multiplicity 3 or higher, or have distinct $\mu$-vectors. \Cref{fig: not two tree-clone children} illustrates that in either of these cases there must be a path (via the red dashed line) to descendants of $v_b$ which do not pass through $v_b$, which implies $v_b$ is not a bridge-node.}

\chris{For the third condition we refer back to \Cref{fig: no other sum for bridge node children}. Let us assume that $v_b$ is a bridge-node but $\mu(\net)$ does contain vectors $\mu(x),\mu(y),\mu(z)$, such that $\mu(z) \nleq \mu(v_b)$, $\mu(z) = \mu(x) + \mu(y)$ and $\mu(x) < \mu(v_b)$. Note that $\mu(z)$ has to belong to a red `outside' node and if it belongs to a reticulation then it also belongs to its tree-node child. Now imagine that the $\mu$-vector of such an outside tree-node is the sum of the $\mu$-vector of a blue `inside' node and another $\mu$-vector. Even if there are multiple pairs (of possible children) whose $\mu$-vectors sum up to the $\mu$-vector of $z$ this would imply $z$ has a child which is not an ancestor of $v_b$.
Yet, there are paths from this child to the leaves below $v_b$. 
This would mean there is a path to a leaf below $v_b$, which passes through $z$ and its child, which does not pass through $v_b$. This contradicts our assumption that $v_b$ is a bridge-node. 
Therefore if $v_b$ is a bridge-node, then  condition 3 must hold.} 

\chris{To prove the other direction, we will show that if $v_b$ is not a bridge-node then one of the above conditions does not hold. If $v_b$ is not a bridge-node, then there is a path from the root to a descendant of $v_b$ which does not pass through $v_b$. Now because $\mu(\rho) \not< \mu(v_b)$ and for any descendant $d$ of $v_b$, with $d \neq v_b$, we have $\mu(d) < \mu(v_b)$. This implies that there is a node $z$ on that path with $\mu(z) \not< \mu(v_b)$ which is the parent of a node $x$ with $\mu(x) < \mu(v_b)$ (because $v_b$ is not a bridge-node this does not necessarily mean that $x$ is itself a descendant of $v_b$ but that does not matter). There are two cases to consider: either $\mu(z) = \mu(v_b)$ or $\mu(z) \nleq \mu(v_b)$.}

\begin{figure}[h]
    \centering
    \includegraphics[width=0.7\linewidth]{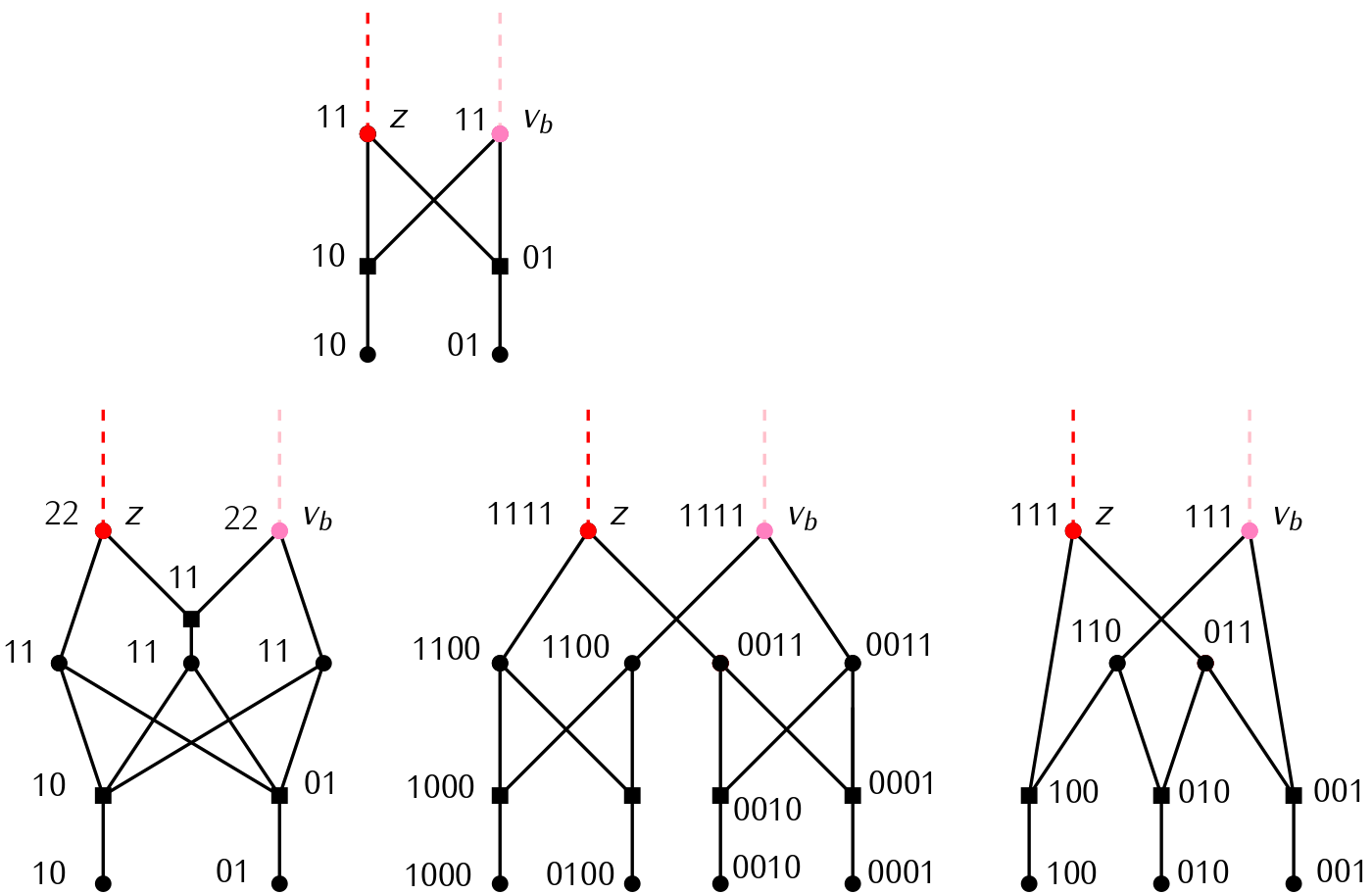}
    \caption{Four networks in which $z$ and $v_b$ are tree-clones with $\mu(z) = \mu(v_b)$.}
    \label{fig: not a tree-clone bridge}
\end{figure}

\chris{If $\mu(z) = \mu(v_b)$, then $\mu(v_b)$ is a tree-clone. \Cref{fig: not a tree-clone bridge} shows four networks where $v_b$ and $z$ share all, some, or none of their children. It can be seen that in the top network and the leftmost network where they share one or more of their children. 
Condition 2 of the proposition does not hold because the multiplicities of the $\mu$-vectors of the children are too high. In the third and fourth cases, they do not share any of their children. In the middle network, the children share the same pair of $\mu$-vectors and therefore the multiplicities are too high, again breaking condition 2. And in the rightmost network, there are two distinct pairs of $\mu$-vectors which sum up to the value of $\mu(v_b)$ and $\mu(z)$, which means that the first condition does not hold.}

\begin{figure}[ht]
    \centering
    \includegraphics[width=0.7\linewidth]{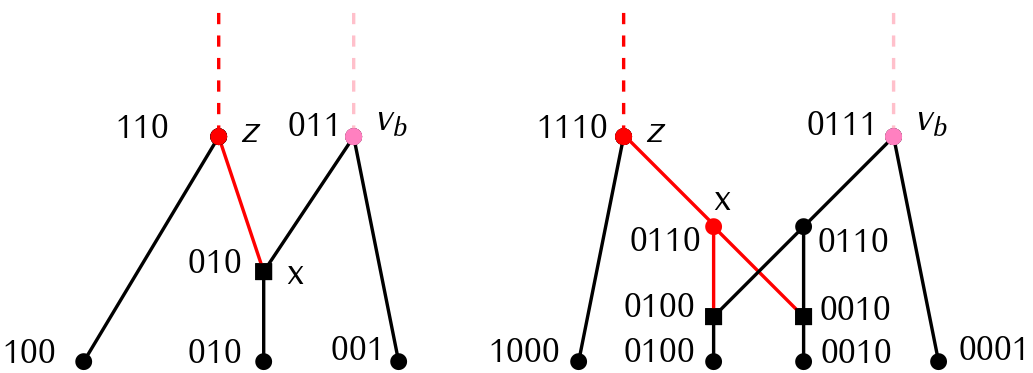}
    \caption{Two networks where there exists a node $z$ with $\mu(z) \nleq \mu(v_b)$ with a child $x$ such that $\mu(x) < \mu(v_b)$.}
    \label{fig: not a bad path bridge}
\end{figure}

\chris{Finally if $\mu(z) \nleq \mu(v_b)$ then condition 3 does not hold because $\mu(z) = \mu(x) + \mu(y)$ for some node $x$ with $\mu(x) < \mu(v_b)$. \Cref{fig: not a bad path bridge} shows two such networks. The second network shows that $x$ does not have to be a descendant of $v_b$ in this case. We have now shown that if $v_b$ is not a bridge node then one of the conditions in \Cref{thm: bridge-node mu-rep} does not hold. We can conclude that the conditions in \Cref{thm: bridge-node mu-rep} are both necessary and sufficient for $\mu(v_b)$ to belong to a bridge-node.}
\end{proof}

\autoref{thm: bridge-node mu-rep} shows that for two semi-binary stack-free networks $\net_1$ and $\net_2$, if $\mu(\net_1) = \mu(\net_2)$, then $\mu(v_b)\in\mu(\net_1)$ belongs to a bridge-node in~$\net_1$ if, and only if it belongs to a bridge-node in~$\net_2$.
In other words, the same $\mu$-vectors will belong to bridge-nodes in both networks. 
While $\mu(\rho)$ also satisfies the conditions of \autoref{thm: bridge-node mu-rep}, a $\mu$-vector belonging to a bridge-node in $\net_1$ will never belong to the root in $\net_2$, since we can identify $\mu(\rho)$ as we noted in \autoref{obs:root mu vector}.
\autoref{lem: lowest ret above bridge}
shows that if $\mu(v_b)$ belongs to a bridge-node in both networks, then the same $\mu$-vector belongs to the lowest reticulation above that bridge-node in both networks. 
Then, by \autoref{obs: strongly stable}, these reticulations will be stable with respect to the same leaves. 
This leads us to the following theorem.

\begin{theorem} \label{thm: strong ret vis in-degrees}
    Let $\net_1$ and $\net_2$ be two strongly reticulation-visible semi-binary networks, where $\mu(\net_1) = \mu(\net_2)$. Then, $\mux(\net_1) = \mux(\net_2)$.
\end{theorem}
\begin{proof}
    As $\net_1$ and $\net_2$ are strongly reticulation-visible, there is a tree-path to a bridge-node from the child of each reticulation. As mentioned above, the same $\mu$-vectors belong to bridge-nodes in both networks, and the same $\mu$-vectors belong to the lowest reticulation above those bridge-nodes. Therefore, the same $\mu$-vectors in both $\mu$-representations will belong to reticulations which are stable for the same set of leaves. To make this more clear, note the following. By \autoref{lem: tree-clones not stable}, $\mu$-vectors which belong to stable nodes cannot belong to tree-clones. Furthermore, by \autoref{cor: stable mult limit}, $\mu$-vectors with multiplicity greater than 2 do not belong to stable nodes. Moreover, as $\net_1$ and $\net_2$ are strongly reticulation-visible, all the reticulations in both networks are stable. Therefore, the $\mu$-vectors with multiplicity greater than 2 in $\mu(\net_1)$ and $\mu(\net_2)$ do not belong to reticulations in either network. 
    By combining \autoref{thm: bridge-node mu-rep} and \autoref{lem: lowest ret above bridge}, we can obtain 
    the set of $\mu$-vectors with multiplicity 2
    that is in bijection with the set of reticulations in both networks. 
    In conclusion, the same $\mu$-vectors belong to reticulations in $\net_1$ and $\net_2$, and they are stable for a common set of leaves. Thus, by \autoref{thm: indeg ret vis}, $\mux(\net_1) = \mux(\net_2)$.
\end{proof} 

\subsection{Encoding strongly reticulation-visible orchard semi-binary networks}

We now combine \autoref{thm: orchard encoding}, \autoref{thm: indeg ret vis}, \autoref{thm: bridge-node mu-rep}, and \autoref{lem: lowest ret above bridge} to prove the following main result.

\begin{theorem} \label{thm: strong ret vis encoding}
    Let $\net_1$ and $\net_2$ be two semi-binary stack-free networks with $\mu(\net_1) = \mu(\net_2)$. Let $\net_1$ be strongly reticulation-visible and orchard. Then, $\net_1 \cong \net_2$.
\end{theorem}
\begin{proof}
    Because $\net_1$ is orchard, it contains no tree-clones. Therefore, each $\mu$-vector in $\mu(\net_1)$ has multiplicity at most 2. And the subset of $\mu(\net_1)$ of $\mu$-vectors with multiplicity 2 is exactly the set of $\mu$-vectors which belong to reticulations in $\net_1$. Furthermore, because $\net_1$ is strongly reticulation-visible, these $\mu$-vectors belong to reticulations which are lowest above some bridge in $\net_1$. By \autoref{thm: bridge-node mu-rep}, the same $\mu$-vectors belong to bridge-nodes in $\net_1$ and $\net_2$. And by \autoref{lem: lowest ret above bridge}, every lowest reticulation - bridge-node pair is preserved in $\net_2$. As both $\mu$-representations do not contain vectors with multiplicity greater than 2, there are no other $\mu$-vectors belonging to reticulations in $\net_2$. Note that each $\mu$-vector with multiplicity 2 belongs to exactly one reticulation in $\net_1$. Therefore, $\net_1$ and $\net_2$ have the same reticulation set $R$ and each reticulation is stable for a common set of leaves in both networks. This means that, by \autoref{thm: indeg ret vis}, $\mux(\net_1) = \mux(\net_2)$. Therefore, by \autoref{thm: orchard encoding}, $\net_1 \cong \net_2$.
\end{proof} 

Recall that in \autoref{subsec:mux_metric}, we defined a metric for the class of semi-binary stack-free orchard networks using modified $\mu$-representations. 
Here, we define an analogous result using \autoref{thm: strong ret vis encoding}.
Let us define the \emph{$\mu$-distance} on networks~$\net_1$ and~$\net_2$ by taking the cardinality of the symmetric difference of $\mu$-representations, i.e., $d_{\mu}(\net_1, \net_2) = |\mu(\net_1) \triangle \mu(\net_2)|$.
By \autoref{thm: strong ret vis encoding}, this is a metric on the class of semi-binary strongly reticulation-visible orchard networks.


\section{Conclusion and Discussion} \label{sec: conclusion & disc}
In this section we will outline and discuss the main results, some of which are displayed in \autoref{tab: results}. 

\begin{table} [ht]
  \centering
  \begin{tabular}{ccccc}
    \arrayrulecolor{CadetBlue}
    \toprule
    \rowcolor{MidnightBlue!20} Given & class $\net_1$ & class $\net_2$ & Result & Theorem \\
    \midrule
    $\mux(\net_1) = \mux(\net_2)$ & SB SF orchard & SB SF & $\net_1 \cong \net_2$ & \autoref{thm: orchard encoding} \\
    \addlinespace[3pt]
    \rowcolor{MidnightBlue!5} $\mu(\net_1) = \mu(\net_2)$ & SB SRV & SB  SRV & $\mux(\net_1) = \mux(\net_2)$ & \autoref{thm: strong ret vis in-degrees}\\
    \addlinespace[3pt]
    $\mu(\net_1) = \mu(\net_2)$ & SB SRV orchard & SB SF & $\net_1 \cong \net_2$ & \autoref{thm: strong ret vis encoding} \\
    \bottomrule
  \end{tabular}
  \caption{A table detailing the main results. SB stands for semi-binary, SF stands for stack-free and SRV stands for strongly reticulation-visible. Note that strongly reticulation-visible networks are stack-free. The first row reads: ``Given two networks $\net_1$ and $\net_2$ with the same modified $\mu$-representation, if $\net_1$ is semi-binary stack-free orchard and $\net_2$ is semi-binary stack-free, then they are isomorphic''.}
  \label{tab: results}
\end{table}
We have shown that semi-binary stack-free orchard networks are encoded in the space of semi-binary stack-free networks by a modified $\mu$-representation (\autoref{thm: orchard encoding}). This modified $\mu$-representation, 
\leo{contains}
the same path multiplicity vectors as the standard $\mu$-representation as originally proposed by Cardona et al. in \cite{cardona2008comparison}, but \leo{additionally includes}
the in-degrees of nodes. 
With this result we have shown that 
the cardinality of the symmetric difference of the $\mux$-representations is a metric for the class of semi-binary stack-free orchard networks.
We have also shown that this encoding result does not extend to non-binary stack-free orchard networks 
(\autoref{thm: non-binary orchard not encoded}) \leo{even if the outdegrees are also included in the modified $\mu$-representation}.\medskip


We proposed the class of strongly reticulation-visible networks, as the class of networks where for each reticulation there is a tree-path from its child to a bridge.
For this class we proved that for any two networks with the same $\mu$-representation, the $\mu$-vectors belong to nodes with equal in-degrees. Therefore, they have the same modified $\mu$-representation (\autoref{thm: strong ret vis in-degrees}). Finally, we concluded that strongly reticulation-visible semi-binary stack-free orchard networks are encoded in the class of semi-binary stack-free networks by their $\mu$-representation (\autoref{thm: strong ret vis encoding}). 
This means that the cardinality of the symmetric difference of the $\mu$-representation is a metric for the class of semi-binary strongly reticulation-visible orchard networks.\medskip

We now give potential future research directions.
All of the following points are elaborated on in~\cite{Reichling2023}.
In \cite{cardona2024comparison} Cardona et al. propose an extended $\mu$-representation for encoding binary orchard networks. We wonder if this can be extended to encode semi-binary orchard networks.
In another direction, we build on the results of \Cref{sec: mu stabilitiy}.
We showed that the bridge-nodes and stability of reticulations in semi-binary strongly reticulation-visible networks can be identified from their $\mu$-representations.
Can we do the same for the non-binary variant?
Finally, 
the original problem of encoding semi-binary stack-free orchard networks using $\mu$-representations
remains open but perhaps our results here can illuminate the next steps towards a proof.

\bibliographystyle{plain}
\newcommand{\etalchar}[1]{$^{#1}$}

\clearpage

\appendix
\section{Proof of \texorpdfstring{\Cref{lem: equivalence mu reduction}}{}}
\begin{proof} 
Let $\net' = (V',E')$ be the network generated from $\net$ by reducing \pair{} and let $\mux'(\net)$ be the multiset generated from $\mux(\net)$ by reducing \pair{}. We will show that $\mux(\net') = \mux'(\net)$.\medskip

First let us assume that \pair{} is a cherry in $\net$. Let $p_{ab}$ be the parent of $a$ and $b$ in $\net$. In this case we have to show that $\mux(\net')$ contains exactly the $\mux$-vectors $(\mux(v)_i)_{i\in (\{0\}\cup X\setminus \{b\})}$ for each $\mux(v) \in \mux(\net)\setminus\{\mux(b),\mux(p_{ab})\}$, where $\mux(p_{ab}) = [1] \oplus (\mu(a) + \mu(b))$. First note that reducing \pair{} in $\net$ does not change the number of paths to any leaf but $b$ for any node which is in both $\net$ and $\net'$.
Moreover, because $\net'$ is a network on leaf set $X \setminus \{b\}$, for each node $v$ which is in both networks, $\mux(\net')$ contains $(\mux(v)_i)_{i\in \{0\}\cup X\setminus \{b\}}$. Furthermore, only the nodes $b$ and $p_{ab}$ are removed, when reducing \pair{} in $\net$. Therefore $V' = V\setminus\{b,p_{ab}\}$ and $p_{ab}$ is a tree-node with children $a$ and $b$ in $\net$, which means $\mux(p_{ab}) = [1] \oplus (\mu(a) + \mu(b))$ in $\mux(\net)$. This shows the claim and therefore $\mux'(\net) = \mux(\net')$.\medskip

Now let us assume that \pair{} is a simple reticulated cherry in $\net$, where $b$ has a reticulation parent $p_b$ and $p_a$ is the parent of $a$. To show that $\mux'(\net) = \mux(\net')$, we have to show that $\mux(\net')$ contains the vector $\mux(a)$ and for each $\mux(v)$ in $\mux(\net) \setminus \{\mux(p_a), \mux(p_b), \mux(a)\}$, $\mux(\net')$ contains a vector $\mux'(v)$ such that $\mux'(v)_i = \mux(v)_i$ for $i \in X \setminus \{b\}$ and $\mux'(v)_b = \mux(v)_b - \mux(v)_a$. Note that in reducing \pair{} in $\net$ the nodes $p_a$ and $p_b$ are suppressed therefore $V' = V \setminus \{p_a, p_b\}$. Furthermore, for each node in $V'$ the number of paths to any leaf other than $b$ in $\net'$ is equal to the number of paths to that leaf in $\net$. However, the paths to leaf $b$ which include $p_a$ in $\net$ are not present in $\net'$ because of the removal of the edge $p_ap_b$. For any node other than $a$ the number of paths to leaf $b$ which pass through $p_a$ in $\net$ is equal to the number of paths to leaf $a$, but there are no paths from $a$ to $b$ in either network. Therefore, $\mux(a)$ is an element of $\mux(\net')$ and for each other node $v$ which is in both networks $\mux(\net')$ contains the vector $\mux'(v)$  such that $\mux'(v)_i = \mux(v)_i$ for $i \in X \setminus \{b\}$ and $\mux'(v)_b = \mux(v)_b - \mux(v)_a$. Thus proving the claim.\medskip

Finally, let us assume that \pair{} is a complex reticulated cherry in $\net$. The only difference as compared to the case where \pair{} was a simple reticulated cherry is the fact that when reducing \pair{} in $\net$ the parent $p_b$ of $b$ is not suppressed but its in-degree is lowered by 1. Therefore, $\mux(p_b)$ is not removed from $\mux(\net)$ but $\mux(p_b)_0$ is lowered by 1.
All other arguments still hold, thus also in this case $\mux(\net') = \mux'(\net)$.
\end{proof}

\section{Formal proof of \texorpdfstring{\Cref{thm: bridge-node mu-rep}}{}}
We use the following lemmas (\autoref{lem: mult of bridge-node children} to \autoref{lem: no other parent of bridge descendants}) to show how to determine whether a non-unit $\mu$-vector belongs to a bridge-node (\autoref{thm: bridge-node mu-rep}).

\begin{lemma} \label{lem: mult of bridge-node children}

    If $v_b$ is a bridge-node and $c_1$ and $c_2$ are its children, then either the $\mu$-vectors of $c_1$ and $c_2$ are distinct and at most one has multiplicity greater than 1 in $\mu(\net)$ or they are equal with multiplicity exactly 2.
\end{lemma}

\begin{proof} We will show the lemma holds by proving the contrapositive: if the $\mu$-vectors of the two children of a tree node are distinct with multiplicity greater than 1 or are equal with multiplicity greater than 2, then their parent is not a bridge-node. This statement can be broken down into five unique cases (there are six in total, but two of the subcases, namely the first and the fourth bullet points, can be dealt with simultaneously).
If the $\mu$-vectors of two nodes $c_1,c_2$, who are children of the same node, are distinct with multiplicity greater than 1, then one of the following holds:
\begin{itemize}
    \item $c_1$ and $c_2$ are reticulations,
    \item $c_1$ and $c_2$ are tree-clones with different $\mu$-vectors,
    \item $c_1$ is a reticulation and $c_2$ is a tree-clone with a different $\mu$-vector.
\end{itemize}
If their $\mu$-vectors are equal with multiplicity greater than 2, one of these must hold:
\begin{itemize}
    \item $c_1$ and $c_2$ are reticulations,
    \item $c_1$ and $c_2$ are tree-clones with the same $\mu$-vector and there exists a third tree-clone $c_3$ with the same $\mu$-vector as $c_1$ and $c_2$,
    \item $c_1$ is a tree-clone and $c_2$ is a reticulation with the same $\mu$-vector.
\end{itemize}
Given a tree-node $v_b$, let us first assume $v_b$ has two children $c_1$ and $c_2$ which are reticulations. Because $c_1$ and $c_2$ are reticulations they must have at least one other parent besides $v_b$. Neither $c_1$ is the parent of $c_2$ nor is $c_2$ the parent of $c_1$ because the network is stack-free. Note that, not both $c_1$ and $c_2$ can have a parent which is a descendant of the other one, because in that case there would be a cycle, from $c_1$ to the parent of $c_2$ to $c_2$ to the parent of $c_1$ to $c_1$. Furthermore, all descendants of $v_b$ except $v_b$ itself are descendants of $c_1$ or $c_2$. Therefore, either $c_1$ or $c_2$ has a parent which is not a descendant of $v_b$. But then there would be a path from the root to a leaf below $v_b$ via this parent, which does not pass through $v_b$. Thus, by \autoref{obs: bridges are stable_1}, this means $v_b$ is not a bridge-node.\medskip

For the second case let us assume $v_b$ has two children $c_1$ and $c_2$ which are tree-clones, with different $\mu$-vectors. In this case there exists a tree-clone $c_1'$ with $\mu(c_1') = \mu(c_1)$ which has a parent which is not $v_b$. And there must be a tree-clone $c_2'$ with $\mu(c_2') = \mu(c_2)$, which has a parent which is not $v_b$. It cannot be the case that the parent of $c_1'$ is a descendant of $c_2$ and the parent of $c_2'$ is a descendant of $c_1$. Because in that case, $\mu(c_1) = \mu(c_1') \leq \mu(c_2)$ and $\mu(c_2) = \mu(c_2') \leq \mu(c_1)$. Which means $\mu(c_1) = \mu(c_2)$ which contradicts our assumption. Nor can the parent of $c_1'$ be a descendant of $c_1$, because then $\mu(c_1') < \mu(c_1)$, which contradicts our assumption. The same holds for $c_2'$ and $c_2$. This means that either, $c_1'$ or $c_2'$ must have a parent which is not a descendant of $v_b$. W.l.o.g. we can assume $c_1'$ has a parent which is not a descendant of $v_b$. Note that, $\mu(c_1') < \mu(v_b)$ so there must be paths from $c_1'$ to leaves below $v_b$. This means that there is a path from the root to a leaf below $v_b$ via $c_1'$ which does not visit $v_b$. Thus, $v_b$ is not a bridge-node.\medskip

Third, let us assume $v_b$ has one child $c_1$ which is a reticulation and a child $c_2$ which is a tree-clone, such that $\#\mu(c_1) \geq 2$, $\#\mu(c_2) \geq 2$ and $\mu(c_1) \neq \mu(c_2)$. If $c_1$ is not a descendant of $c_2$ then $c_1$ has another parent which is not a descendant of $v_b$. This means there must be a path from the root to a leaf below $v_b$ via this parent, which does not pass through $v_b$. In this case $v_b$ is not a bridge-node. Alternatively, let us assume $c_1$ is a descendant of $c_2$. Then there must exist tree-clone $v$ with $\mu(c_2) = \mu(v)$ and $\mu(v) > \mu(c_1)$, with a parent which is not a descendant of $v_b$. Because $\mu(v_b) > \mu(c_2) = \mu(v)$ there must be a path from $v$ to a leaf below $v_b$ and therefore, there must be a path from the root to a leaf below $v_b$ via $v$, which does not pass through $v_b$. This means $v_b$ is not a bridge-node.\medskip

Fourth, let us assume that $v_b$ has two children $c_1$ and $c_2$, which are tree-clones with $\mu(c_1) = \mu(c_2)$ and there exists another tree-clone $c_3$ with $\mu(c_3) = \mu(c_1) = \mu(c_2)$. Note that the parent of $c_3$ cannot be a descendant of $c_1$ or $c_2$, because $c_1$ and $c_2$ are tree-nodes and this would mean either $\mu(c_3) < \mu(c_1)$ or $\mu(c_3) < \mu(c_2)$, which we have assumed is not the case. Nor can the parent of $c_3$ be $v_b$, because $v_b$ already has 2 children. Therefore, $c_3$ must have a parent which is not a descendant of $v_b$ and, by the same argument as before, this implies $v_b$ is not a bridge-node.\medskip

Finally, let us assume $v_b$ has two children, $c_1$ which is a tree-node and $c_2$ which is a reticulation, with $\mu(c_2) = \mu(c_1)$. Then, $c_2$ must have at least one other parent, which is not $v_b$. This parent cannot be a descendant of $c_1$, because $c_1$ is a tree-node and therefore this would imply that $\mu(c_2) < \mu(c_1)$, which contradicts our assumption. Therefore, we can assume $c_2$ has a parent which is not a descendant of $v_b$. Which again implies that $v_b$ is not a bridge-node. 
\end{proof}

\begin{lemma} \label{lem: bridge-node non ambiguous}
    If $\mu(v_b)$ belongs to a bridge-node which is not a leaf then $\mu(\net)$ contains exactly one pair of vectors $\mu(x),\mu(y)$ such that $\mu(v_b) = \mu(x) + \mu(y)$.
\end{lemma}
\begin{proof}
    We will show the lemma is true with a proof by contradiction, by showing that there cannot be a second pair. Assume $\mu(v_b)$ is a bridge-node which is not a leaf and $\mu(\net)$ contains at least two pairs $\mu(x),\mu(y)$ and $\mu(k),\mu(\ell)$ such that $\mu(v_b) = \mu(x) + \mu(y) = \mu(k) + \mu(\ell)$. W.l.o.g. we can assume $\mu(x)$ and $\mu(y)$ belong to the children $x,y$ of $v_b$. By \autoref{lem: mult of bridge-node children}, we know that
    $\mu(k) \neq \mu(x)$, because otherwise $\mu(\ell) = \mu(y)$, in which case both $\#\mu(x) \geq 2$ and $\#\mu(y)\geq 2$, and if $\mu(x) = \mu(y)$ then $\#\mu(x) \geq 4$. The same goes for $\mu(k) \neq \mu(y)$, $\mu(\ell) \neq \mu(x)$ and $\mu(\ell) \neq \mu(y)$. Note also that none of $x,y,k,\ell$ are ancestors of $v_b$ because their $\mu$-vectors are each lower than $\mu(v_b)$. Now note that, by \autoref{obs: bridges are stable_2}, $\mu(v_b) > \mu(k)$ and $\mu(v_b) > \mu(\ell)$ implies they are descendants of $v_b$. However, as we mentioned $x,y$ are the children of $v_b$. This means $k,\ell$ must be descendants of $x,y$, which are not equal to $x,y$, because their $\mu$-vectors differ. From this it follows that $\mu(x) + \mu(y) > \mu(k) + \mu(\ell)$. This contradicts our assumption that the sums were equal.
\end{proof}

\begin{lemma} \label{lem: mult of tree-clone children}
    If $\mu(v_t) \in \mu(\net)$ belongs to a tree-clone, then one of the following is true:
    \begin{itemize}
        \item $\mu(\net)$ contains exactly one pair $\mu(x),\mu(y)$ such that $\mu(v_t) = \mu(x) + \mu(y)$ and either~$\mu(x) \neq \mu(y)$, in which case $\#\mu(x)\geq 2$ and $\#\mu(y)\geq 2$, or $\mu(x) = \mu(y)$, in which case $\#\mu(x) \geq 4$.
        \item $\mu(\net)$ contains at least one more pair $\mu(k),\mu(\ell)$ which is distinct from $\mu(x),\mu(y)$, such that $\mu(v_t) = \mu(k) + \mu(\ell)$. 
    \end{itemize}
\end{lemma}

\begin{proof}
   If $\mu(v_t) \in \mu(\net)$ belongs to a tree-clone, then there are at least two tree-nodes $v_1$ and $v_2$ with $\mu$-vector equal to $\mu(v_1)$. By \autoref{lem: stack free leaf mults}, unit-vectors do not belong to tree-clones, therefore $v_1$ and $v_2$ are not leaves. This means $v_1$ and $v_2$ each have two children. Let $c_1$ and $c_2$ be the children of $v_1$, and let $c_3$ and $c_4$ be the children of $v_2$. Note that $c_1$ must be distinct from $c_2$, and $c_3$ must be distinct from $c_4$, because phylogenetic networks do not contain parallel edges. We then have: $\mu(c_1) + \mu(c_2) = \mu(v_1) = \mu(v_2) = \mu(c_3) + \mu(c_4)$. The statement can be broken down into three unique cases.\medskip

    If the children of $v_1$ are the same as the children of $v_2$ then $c_1$ and $c_2$ both have two parents and are therefore reticulations. In that case, $\#\mu(c_1) \geq 2$ and $\#\mu(c_2) \geq 2$, and if $\mu(c_1) = \mu(c_2)$, then $\#\mu(c_1)\geq 4$, because $c_1$ cannot have the same child as $c_2$ because the network is stack-free.\medskip

    If $v_1$ and $v_2$ share only one child, say $c_2 = c_3$, then $c_1$ and $c_4$ are distinct nodes. In this case, $c_2$ is a reticulation, and $\#\mu(c_2) \geq 2$. Then, either $\mu(c_1) = \mu(c_4) \neq \mu(c_2)$, so that $\#\mu(c_1) \geq 2$ as well. Or $\mu(c_1) = \mu(c_4) = \mu(c_2)$, which implies $\#\mu(c_1) \geq 4$, because neither $c_1$ nor $c_4$ can be the child of $c_2$, because the network is stack-free.\medskip

    Finally, if all nodes $c_1,c_2,c_3$ and $c_4$ are distinct from each other. Then either, the sets $\{\mu(c_1),\mu(c_2)\}$ and $\{\mu(c_3),\mu(c_4)\}$ are not the same set, in which case $\mu(\net)$ contains two distinct pairs, whose sum is $\mu(v_1)$. Or they are the same set, in which case we can say w.l.o.g. that $\mu(c_1) = \mu(c_3)$, which implies $\mu(c_2) = \mu(c_4)$. Which means that both $\#\mu(c_1) \geq 2$ and $\#\mu(c_2) \geq 2$, and if $\mu(c_1) = \mu(c_2)$ then $\#\mu(c_1) \geq 4$.
\end{proof}

\begin{lemma} \label{lem: bridge mu-vector factorization}
    Let $\mu(v_b) \in \mu(\net)$ belong to a bridge-node $v_b$. Let $I_b$ be the set of leaves below $v_b$, so that $\mu(v_b)_i = 0$ for $i \notin I_b$. For any $\mu(x) \in \mu(\net)$ with $\mu(x) \not< \mu(v_b)$, it holds that $\mu(x) = P_{xv_b}\mu(v_b) + \mu'(x)$, where $P_{xv_b}$ is a non-negative integer, equal to the number of paths from $x$ to $v_b$, and $\mu'(x)$ is a $\mu$-vector such that $\mu'_i(x) = 0$ for $i \in I_b$.
\end{lemma}
\begin{proof}
    Let $\mu(v_b) \in \mu(\net)$ belong to a bridge-node. For any $\mu(x) \in \mu(\net)$ with $\mu(x) \not< \mu(v_b)$, by \autoref{obs: bridges are stable_2},  $\mu(x)$ does not belong to a descendant of the bridge-node with $\mu$-vector $\mu(v_b)$. Therefore, all paths from $x$ to leaves below $v_b$ must visit $v_b$. Thus, if $a$ is a leaf below $v_b$, then any path from $x$ to $a$ is the composition of a path from $x$ to $v_b$ and a path from $v_b$ to $a$. It follows that for each such leaf $a$, $\mu(x)_a = P_{xv_b}\mu(v_b)_a$. From this it follows that $\mu(x) = P_{xv_b}\mu(v_b) + \mu'(x)$, where $\mu'(x)$ contains only the paths to leaves not below $v_b$, and therefore $\mu'(x)$ is a $\mu$-vector such that $\mu'_i(x) = 0$ for $i \in I_b$.
\end{proof}

\begin{lemma} \label{lem: non-ambiguous bridge descendants}
    Let $\mu(v_b) \in \mu(\net)$ belong to a bridge-node $v_b$. Then, for any $\mu(x), \mu(y) \in \mu(\net)$ such that $\mu(x) < \mu(v_b)$, $\mu(\net)$ does not contain $\mu(k),\mu(\ell)$ with $\mu(k) \not< \mu(v_b)$ and $\mu(\ell) \not< \mu(v_b)$, such that $\mu(k) + \mu(\ell) = \mu(x) + \mu(y)$.
\end{lemma}
\begin{proof}
    We will show this with a proof by contradiction. Let us assume that $\mu(v_b) \in \mu(\net)$ belong to a bridge-node $v_b$ and $\mu(\net)$ contains some $\mu(x),\mu(y)$ such that $\mu(x) < \mu(v_b)$. Now let us assume, $\mu(\net)$ contains $\mu(k),\mu(\ell)$ with $\mu(k) \not< \mu(v_b)$ and $\mu(\ell) \not< \mu(v_b)$, such that $\mu(k) + \mu(\ell) = \mu(x) + \mu(y)$. First note, that if $\mu(y) \leq \mu(v_b)$, then $\mu(y)$ belongs to a descendant of $v_b$ and therefore $\mu(k) + \mu(\ell) = \mu(x) + \mu(y) \leq \mu(v_b)$. But this contradicts our assumption that $\mu(k) \not< \mu(v_b)$. Therefore, we can assume that $\mu(y) \nleq \mu(v_b)$.
    Then, by \autoref{lem: bridge mu-vector factorization}:
    \begin{align*}
        \mu(k) &= P_{kv_b}\mu(v_b) + \mu'(k) \\
        \mu(\ell) &= P_{\ell v_b}\mu(v_b) + \mu'(\ell)
    \end{align*}
    and
    \begin{align*}
        \mu(y) = P_{yv_b}\mu(v_b) + \mu'(y).
    \end{align*}
    Furthermore, from $\mu(x) + \mu(y) = \mu(k) + \mu(\ell)$ it follows:
    \begin{align*}
        \mu(x) &= \mu(k) + \mu(\ell) - \mu(y)\\
            &= P_{kv_b}\mu(v_b) + \mu'(k) + P_{\ell v_b}\mu(v_b) + \mu'(\ell) - [P_{yv_b}\mu(v_b) + \mu'(y)]\\
            &= (P_{kv_b} + P_{\ell v_b} - P_{yv_b})\mu(v_b) + \mu'(k) + \mu'(\ell) - \mu'(y).
    \end{align*}
    Then, from $\mu(x) < \mu(v_b)$ it follows that $P_{kv_b} + P_{\ell v_b} - P_{yv_b} = 0$, and $\mu'(k) + \mu'(\ell) - \mu'(y) = 0$. But then $\mu(x)$ is the zero vector, which is not possible because the zero vector is not a $\mu$-vector and therefore not contained in $\mu(\net)$. Thus we have reached a contradiction.
\end{proof}
\begin{lemma} \label{lem: no other parent of bridge descendants}
    Given a $\mu$-vector $\mu(v_b) \in \mu(\net)$. If $\mu(\net)$ contains $\mu(z) \nleq \mu(v_b)$, $\mu(x) < \mu(v_b)$ and $\mu(y)$, such that $\mu(z) = \mu(x) + \mu(y)$. Then, $\mu(v_b)$ does not belong to a bridge-node.
\end{lemma}
\begin{proof}
    We will argue by contradiction. Assume $\mu(\net)$ contains $\mu(v_b),\mu(z),\mu(x)$ and $\mu(y)$ as described in the lemma and let $\mu(v_b)$ belong to a bridge-node $v_b$. If $\mu(z)$ belongs to a reticulation $r$, then the child of $r$ must be a tree-node with the same $\mu$-vector, so w.l.o.g. we can assume $\mu(z)$ belongs to a tree-node $z$. Because $\mu(z) = \mu(x) + \mu(y)$, we know $z$ is not a leaf, because $\mu(z)$ is not a unit vector. If $\mu(x),\mu(y)$ is the only pair in $\mu(\net)$ which sum up to $\mu(z)$ then $z$ must be the parent of $x$. If there are more pairs in $\mu(\net)$ then by \autoref{lem: non-ambiguous bridge descendants}, each of those pairs must contain at least one $\mu$-vector lower than $\mu(v_b)$. This means that in any case $z$ will have a child with $\mu$-vector lower than $\mu(v_b)$. However, $\mu(z) \nleq \mu(v_b)$ implies that $z$ is not a descendant of $v_b$.  Then there must be a path from the root to a leaf below $v_b$ via $z$ which does not visit $v_b$. By \autoref{obs: bridges are stable_1}, this means that $v_b$ is not a bridge-node. Which contradicts our assumption.
\end{proof}

\begin{proposition*}[\ref{thm: bridge-node mu-rep}]
    Let~$\mu(v_b)$ be a non-unit vector that is not $\mu(\rho)$.
    Then, it belongs to a bridge-node if, and only if, 
    \begin{enumerate}
        \item there is exactly one pair $\mu(k),\mu(\ell) \in \mu(\net)$ such that $\mu(v_b) = \mu(k) + \mu(\ell)$, and
        \item \chris{one of the following holds}
        \begin{enumerate}
            \item \chris{at least one of} $\mu(k)$ and $\mu(\ell)$ \chris{has multiplicity 1} in $\mu(\net)$; or
            \item $\mu(k) = \mu(\ell)$ \chris{and} $\#\mu(k) = 2$, and
        \end{enumerate}
        \item $\mu(\net)$ does not contain vectors $\mu(x),\mu(y),\mu(z)$, such that $\mu(z) \nleq \mu(v_b)$, $\mu(z) = \mu(x) + \mu(y)$ and $\mu(x) < \mu(v_b)$.
    \end{enumerate}
\end{proposition*}
\begin{proof}
    Assume the non-unit vector $\mu(v_b)$ belongs to a bridge-node. Then $\mu(v_b)$ belongs to a tree-node which is not a leaf and, by \autoref{lem: bridge-node non ambiguous}, there exists exactly one pair $\mu(k),\mu(\ell)$ with $\mu(v_b) = \mu(k) + \mu(\ell)$. Furthermore, by \autoref{lem: mult of bridge-node children}, the combined multiplicity of  $\mu(k)$ and $\mu(\ell)$ in $\mu(\net)$ is lower than or equal to 3. Finally, by \autoref{lem: no other parent of bridge descendants}, there does not exist $\mu(z) \nleq \mu(v_b)$ such that $\mu(z) = \mu(x) + \mu(y)$ for $\mu(x) < \mu(v_b)$. This proves the first direction of the biconditional. \medskip

    For the other direction, we will show the inverse holds. Let us assume the non-unit vector $\mu(v_b)$ does not belong to a bridge-node. Then, for each node $v_b$ with $\mu$-vector $\mu(v_b)$ there must be a path from the root to a descendant of $v_b$, which does not visit $v_b$. Note that for any reticulation $r$ with $\mu$-vector equal to $\mu(v_b)$, there must be a tree-node with the same $\mu$-vector and $r$ itself is not a bridge-node. So w.l.o.g. it is enough to show that this holds for tree-nodes with $\mu$-vector $\mu(v_b)$. Note that, for the $\mu$-vector of the root $\mu(\rho) \nleq \mu(v_b)$. Therefore, there must be a node $z$ on the path from the root to a descendant of $v_b$ which does not visit $v_b$, with $\mu(z) \not< \mu(v_b)$, and $z$ is the parent of a node $x$, with $\mu(x) < \mu(v_b)$. Note that $z$ cannot be a reticulation, because then $\mu(v_b) \not> \mu(z) = \mu(x) < \mu(v_b)$.\medskip
    
    
    Now there are two cases we should consider, either ~$\mu(z) = \mu(v_b)$ or ~$\mu(z) \nleq \mu(v_b)$. If ~$\mu(z) = \mu(v_b)$, then $\mu(v_b)$ belongs to a tree-clone. In that case, there is at least one pair $\mu(k),\mu(\ell) \in \mu(\net)$ which belong to the children of a node $v_b$, such that $\mu(v_b) = \mu(k) + \mu(\ell)$. Then, by \autoref{lem: mult of tree-clone children}, either there is more than one such pair, or $\#\mu(k) \geq 2$ and $\#\mu(\ell) \geq 2$, and if $\mu(k) = \mu(\ell)$ then $\#\mu(k) \geq 4$. This violates condition 1 or 2.\medskip
    
    If $\mu(z) \nleq \mu(v_b)$, then $\mu(v_b)$ does not necessarily belong to a tree-clone. Note $z$ is also not a leaf, as a leaf has no children. Therefore $z$ is a tree-node with two children, one of which is $x$. This means there exists a node $y$, the other child of $z$, such that $\mu(z) = \mu(x) + \mu(y)$. This violates condition 3. Now, we have shown one of the three conditions must be false. This proves the inverse statement.
\end{proof}


\end{document}